\documentclass[11pt,letter]{article}
\usepackage[top=1in, bottom=1in, left=1in,right=1in]{geometry}
\usepackage{graphicx} 
\usepackage{mathtools}
\usepackage[pdfstartview=FitH,
pdfpagemode=None,
colorlinks=true]{hyperref}
\usepackage{amsmath,amssymb,amsthm}
\usepackage{cleveref} 

\usepackage{authblk}

\crefformat{equation}{(#2#1#3)}

\usepackage{documentstyle} 
\usepackage{complexityToC} 
\usepackage{basicmathmacros} 
\usepackage{bbm}
\usepackage{nicefrac} 
\newtheorem{fact}[theorem]{Fact}
\newcommand{\spant}{\operatorname{span}}

\usepackage{bbm} 

\hypersetup{urlcolor=[rgb]{0.8, .2, .5},
        linkcolor=[rgb]{0.1,0.1,0.6},
        citecolor=[rgb]{0.2,0.2,0.7}
}
\usepackage[
    backend=biber,
    style=alphabetic,
    maxbibnames=99,
    maxalphanames=99,
  ]{biblatex}
\newrobustcmd*{\citeallauthors}{%
  \AtNextCite{\AtEachCitekey{\defcounter{maxnames}{999}}}%
  \citeauthor}

\usepackage{jiaqi-macros}
\usepackage{turnstile} 

\newcommand{\ks}{\mathsf{ks}}

\date{}
\usepackage[normalem]{ulem}

\begin{document}

\sloppy

\title{Lower Bounds against the Ideal Proof System in Finite Fields\thanks{%
This project was supported by the Engineering and Physical Sciences Research
Council (grant EP/Z534158/1) and the
European Research Council (ERC) under the European Unions Horizon 2020 research and innovation programme (grant agreement 101002742; EPRICOT project).%
    }%
}


\renewcommand\Authsep{\qquad }
\renewcommand\Authand{ }
\renewcommand\Authands{\qquad }
 
\author{Tal Elbaz\thanks{E-mail: t.elbaz22@imperial.ac.uk}}
\author{Nashlen Govindasamy\thanks{E-mail: nashlen.govindasamy@gmail.com}}
\author{Jiaqi Lu\thanks{E-mail: jiaqi.lu23@imperial.ac.uk}}
\author{Iddo Tzameret\thanks{E-mail: iddo.tzameret@gmail.com. \url{https://www.doc.ic.ac.uk/\~itzamere/}}}
\affil{Imperial College London\\ {\normalsize Department of Computing}}

\maketitle

\begin{abstract}
Lower bounds against strong algebraic proof systems and specifically fragments of the Ideal Proof System ($\IPS$), have been obtained in an ongoing line of work.
All of these bounds, however, are proved only over large (or characteristic $0$) fields,\footnote{Except for the \emph{placeholder} lower bound model, where the instance itself lacks small circuits \cite{HLT24}.}
yet finite fields are the more natural setting for propositional proof complexity, especially for progress toward lower bounds for Frege systems such as $\ACZ[p]$-Frege.
This work establishes lower bounds against fragments of $\IPS$ over fixed finite fields.
Specifically, we show that a variant of the knapsack instance studied by \citeallauthors{GHT22} (\begin{footnotesize}FOCS\end{footnotesize}'22) has no polynomial-size $\IPS$ refutation over finite fields when the refutation is multilinear and written as a constant-depth circuit.
The key ingredient of our argument is the recent set-multilinearization result of \citeallauthors{forbes24} (\begin{footnotesize}CCC\end{footnotesize}'24), which extends the earlier result of \citeallauthors{LST21} (\begin{footnotesize}FOCS\end{footnotesize}'21) to all fields, and an extension of the techniques of \citeallauthors{GHT22} to finite fields.
%
We also separate this proof system from the one studied by \citeallauthors{GHT22}.

In addition, we present new lower bounds for read-once algebraic branching program refutations, roABP-$\IPS$, in finite fields, extending results of \citeallauthors{FSTW21} (\begin{footnotesize}Theor.~of~Comput.\end{footnotesize}'21) and \citeallauthors{HLT24}~(\begin{footnotesize}STOC'24\end{footnotesize}).

Finally, we show that any lower bound against any proof system at least as strong as (non-multilinear) constant-depth $\IPS$ over finite fields for \emph{any} instance, even a purely algebraic instance (i.e., not a translation of a Boolean formula or CNF), implies a hard \emph{CNF formula} for the respective $\IPS$ fragment, and hence an $\ACZ[p]$-Frege lower bound by known simulations over finite fields (\citeallauthors{GP14}~(\begin{footnotesize}J.~ACM\end{footnotesize}'18)). 
\end{abstract}

\noindent\textit{Note on independent concurrent work:}
Independently and concurrently with our work, \citeallauthors{BLRS25} \cite{BLRS25} using different arguments, obtained related results for fragments of $\IPS$ over fields of positive characteristic.
Both works establish a lower bound for \emph{constant-depth multilinear $\IPS$} but the field assumptions differ:
\cite{BLRS25} requires the size of the field to grow with the instance, whereas our lower bound holds for any field of constant (or small) positive characteristic.
Hence, in the constant positive characteristic setting our constant-depth multilinear IPS lower bound strictly subsumes theirs as it also holds over any fixed finite field.
We provide a detailed comparison of these and, where relevant, other results of \cite{BLRS25} in \Cref{sec:results-and-technique}.








\section{Introduction}


This work investigates lower bounds against the Ideal Proof System ($\IPS$) over finite fields, motivated by two main considerations. First, existing lower bounds for $\IPS$ have not adequately addressed the case of finite fields. Second, focusing on finite fields---rather than large fields---offers a more natural setting for tackling a central open problem in proof complexity: proving super-polynomial lower bounds against $\ACZ[p]$-Frege.

\subsection{Algebraic Proof Complexity}

Proof complexity studies the size of proofs that certify membership in languages such as UNSAT, the set of unsatisfiable Boolean formulas.
In this context, a proof is a witness that can be verified efficiently, and for UNSAT, such a proof is typically called a refutation.
A central objective of the field is to establish lower bounds against increasingly powerful proof systems, with the overarching goal of demonstrating that no proof system admits polynomial-size refutations for all unsatisfiable formulas.
This approach is often referred to as \emph{Cook's Programme}, following Stephen Cook's suggestion in the 1970s that proof complexity lower bounds could yield insights into fundamental questions in computational complexity, such as the \P\ versus \NP\ problem.
In particular, showing that no proof system can efficiently refute all unsatisfiable formulas would separate \NP\ from \coNP\, and thereby separate \P\ from \NP.

An important thread of proof complexity is to investigate \emph{algebraic proof systems}, which certify that a given set of multivariate polynomials over a field has no common Boolean solution.
Some of the foundational proof systems in this line are the Polynomial Calculus (PC) \cite{CEI96} and its `static' variant, Nullstellensatz \cite{BeameIKPP96}.
In PC, proofs proceed by algebraic manipulation, adding and multiplying polynomials, until deriving the contradiction $1 = 0$.
Contrastingly, in Nullstellensatz,  a proof of the unsatisfiability of a set of axioms, written as polynomial equations $\{f_i(\vx)=0\}$ over a field, is a \emph{single} polynomial identity expressing $1$ as a combination of the axioms, that is:
\begin{equation}
\label{eq:NS-first}
\sum_i g_i(\vx)\cdot f_i(\vx) = 1\,,
\end{equation}
for some polynomials $\{g_i(\vx)\}$.
These systems measure proof size by sparsity, defined as the total number of monomials involved, which makes them comparatively weak.
An alternative way to measure proof size is by  algebraic circuit size. This was suggested initially by Pitassi \cite{Pit97,Pit98}, and further investigated in the work of Grigoriev and Hirsch \cite{GH03} and subsequently Raz and Tzameret \cite{RT06,RT07}, eventually leading to the Ideal Proof System \cite{GP14} described in what follows.

\subsection{Ideal Proof System}
The Ideal Proof System ($\IPS$, for short; \Cref{def:IPS}), introduced by Grochow and Pitassi \cite{GP14}, loosely speaking is the Nullstellensatz proof system where the polynomials $g_i(\bar{x})$ in \Cref{eq:NS-first} are represented by algebraic circuits. Formally, Forbes, Shpilka, Tzameret and Wigderson~\cite{FSTW21} showed that $\IPS$ is equivalent to Nullstellensatz in which the polynomials $g_i$ in Equation \Cref{eq:NS-first} are written as algebraic circuits. In other words, an $\IPS$ refutation of the set of axioms $\left\{f_i(\bar{x})=0\right\}_i$ can be defined similarly to Equation \Cref{eq:NS-first} (here we display explicitly the Boolean axioms $x_j^2-x_j$):
\begin{equation}
\label{eq:intro:IPS}
\sum_i g_i(\bar{x}) \cdot f_i(\bar{x})+\sum_j h_j(\bar{x}) \cdot\left(x_j^2-x_j\right)=1,
\end{equation}
for some polynomials $\left\{g_i(\bar{x})\right\}_i$, where we think of the polynomials $g_i, h_j$ written as algebraic circuits (instead of e.g., counting the number of monomials they have towards the size of the refutation). Thus, the size of the $\IPS$ refutation in Equation (0.2) is $\sum_i \operatorname{size}\left(g_i(\bar{x})\right)+\sum_j \operatorname{size}\left(h_j(\bar{x})\right)$, where $\operatorname{size}(g)$ stands for the (minimal) size of an algebraic circuit computing the polynomial $g$.

When considering algebraic circuit classes weaker than general algebraic circuits, one has to be a bit careful with the definition of $\IPS$. For technical reasons the formalization in \Cref{eq:intro:IPS} does not capture the precise definition of $\IPS$ restricted to the relevant circuit class, rather the fragment which is denoted by $\mathcal{C}$-$\lIPS$ (``LIN" here stands 
for the linearity of the axioms $f_i$ and the Boolean axioms; that is, they appear with power 1). In this work, we focus on $\mathcal{C}$-$\lIPS$ and a similar stronger variant denoted $\mathcal{C}$-$\lbIPS$. 
Throughout the introduction, refutations in the system $\mathcal{C}$-$\lIPS$ are defined as in Equation \Cref{eq:intro:IPS} where the polynomials $g_i, h_j$ are written as circuits in the circuit class $\mathcal{C}$.

Technically, our lower bounds are proved by lower bounding the algebraic circuit size of the $g_i$'s in \Cref{eq:intro:IPS}, namely the products of the axioms $f_i$, and not the products of the Boolean axioms (that is, we ignore the circuit size of the $h_i$'s). For this reason, our lower bounds are slightly stronger than lower bounds on $\mathcal{C}$-$\lIPS$, rather they are lower bounds on the system denoted $\mathcal{C}$-$\lbIPS$ (see \Cref{def:IPS}).


\paragraph{Lower bounds methods and known results.}

\citeallauthors{FSTW21} \cite{FSTW21} introduced two approaches for turning algebraic circuit lower bounds into lower bounds for $\IPS$: the \emph{functional lower bound} method and the \emph{lower bounds for multiples} method.
Of the two, the functional approach has proved more instrumental, proving several concrete proof complexity lower bounds against fragments of $\IPS$.
These include lower bounds for variants of the subset-sum instance against $\IPS$ refutations written as read-once (oblivious) algebraic branching programs (roABPs), depth-$3$ powering formulas, and multilinear formulas \cite{FSTW21}.
A similar method underpinned the \emph{conditional} lower bound against general $\IPS$ established by \citeallauthors{AGHT20} \cite{AGHT20} (leading to \cite{Ale21}).
\citeallauthors{GHT22} \cite{GHT22} combined the functional method with the constant-depth algebraic circuit lower bound result of \citeallauthors{LST21} \cite{LST21}, obtaining constant-depth multilinear $\IPS$ lower bounds.

By contrast, the multiples method has so far matched the functional method only within the weaker \emph{placeholder} model of $\IPS$, where the hard instances themselves do not have small circuits in the fragment under study \cite{FSTW21, AF22}.
Other approaches have emerged as well: the \emph{meta-complexity} approach of \citeallauthors{ST25} \cite{ST25}, which obtains a \emph{conditional} $\IPS$ size lower bound on a self-referential statement; the \emph{noncommutative} approach of \citeallauthors{LTW18} \cite{LTW18} (building on \cite{Tza11-I&C}, which reduced Frege lower bounds to matrix-rank lower bounds but has yet to yield concrete lower bounds; and recent lower bounds against \emph{PC with extension variables} over finite fields of \citeallauthors{IMP23_ccc} \cite{IMP23_ccc} (building on \cite{Sok20} and improved by \cite{DMM24_sat}) which can be considered as a fragment of $\IPS$ sitting between depth-$2$ and depth-$3$.

The functional lower bound method was further investigated by \citeallauthors{HLT24} \cite{HLT24}.
There, Nullstellensatz degree lower bounds for symmetric instances and vector invariant polynomials were established, which were then lifted to yield $\IPS$ size lower bounds for the roABP and multilinear formula fragments of $\IPS$.
With invariant polynomials, the bounds hold over \emph{finite fields}, though within the \emph{placeholder} model.
Building on recent advances in constant-depth algebraic circuit lower bounds from \cite{AGK+23}, they extend \cite{GHT22} to constant-depth $\IPS$ refutations computing polynomials with $O(\log \log n)$ individual degree.
Finally, they observe a barrier in that the functional method cannot yield lower bounds for any Boolean instance against sufficiently strong proof systems like constant-depth $\IPS$. 

\subsection{IPS over Finite Fields}

$\IPS$ lower bounds have so far been obtained almost exclusively over fields of large characteristic.
Finite fields, however, are a more natural setting for propositional proof complexity, particularly for the long-standing open problem of establishing super-polynomial lower bounds for $\ACZ[p]$-Frege.
This proof system operates with constant-depth propositional formulas equipped with modulo $p$ counting gates, where $p$ is a prime.
\citeallauthors{GP14} \cite{GP14} showed that constant-depth $\IPS$ refutations over $\FF_p$ simulate $\ACZ[p]$-Frege, thus obtaining lower bounds against $\IPS$ over finite fields provides a concrete approach to settle the problem of obtaining lower bounds against $\ACZ[p]$-Frege.
Although lower bounds against $\ACZ[p]$-Frege are sometimes thought to be within reach of current techniques, especially given existing lower bounds against both $\ACZ$-Frege and $\ACZ[p]$ circuits, this problem and the problem of obtaining lower bounds against constant-depth $\IPS$ over $\FF_p$ remain elusive.

$\IPS$ lower bounds over finite fields face additional challenges that are not present in the characteristic $0$ setting. 
A recurring obstacle is provided by Fermat's little theorem:
for a nonzero $a \in \FF_p$, $a^{p-1} = 1$ in $\FF_p$.
More generally, if $\FF$ is a finite field of size $q$, then for a nonzero $a \in \FF$, $a^{q-1} = 1$.
Hence, if a polynomial $f \in \FF[\vx]$ admits no satisfying Boolean assignment, then $(f(\vx))^{(q-2)} f(\vx) = (f(\vx))^{(q-1)} = 1$ over Boolean assignments.
The functional lower bound method of \cite{FSTW21} requires a lower bound on the size of circuits computing $g(\vx)$ such that $g(\vx) f(\vx) = 1$ over Boolean assignments, hence requires a lower bound on the size of $(f(\vx))^{(q-2)}$.
Thus we must simultaneously ensure that the hard instance $f(\vx)$ is easily computed by the subsystem of $\IPS$ under consideration while $(f(\vx))^{(q-2)}$ is hard in that same subsystem.
While this is not possible for proof systems closed under constant multiplication of polynomials, including certain constant-depth $\IPS$ subsystems (for example the one studied in \cite{HLT24} which considered $\log\log n$ individual degree refutations), for the multilinear constant-depth $\IPS$ subsystem considered in \cite{GHT22}, it is indeed possible, even for constant $q$.

\subsection{Our Results}
\label{sec:results-and-technique}
\subsubsection{Bounds for Constant-depth IPS over Finite Fields}
Our first contribution establishes a super-polynomial lower bound for constant-depth $\lbIPS$ refutations \emph{over finite fields}. 
As mentioned above, $\lbIPS$ is the Nullstellensatz proof system \Cref{eq:intro:IPS} whose refutations are algebraic circuits (see \Cref{def:IPS}).
This result is the finite field analogue of \cite{GHT22}, which was proved over characteristic $0$ fields.
Our hard instance is the knapsack mod $p$ polynomial $\ks_{w,p}$, a variant of the knapsack polynomial $\ks_w$ used in their work.
\begin{theorem}[Informal; see \Cref{thm:ips-modp-main}]
\label{thm:ips-modp-main-intro}
    Let $p \geq 5$ be a prime, and let $\FF$ be a field of characteristic $p$.
    Every constant-depth multilinear \lbIPS\ refutation over $\FF$ of the knapsack mod $p$ instance $\ks_{w,p}$ requires super-polynomial size.
\end{theorem}
The proof in \cite{GHT22} combines two main ingredients: first, the methods used by \citeallauthors{LST21} \cite{LST21} to establish super-polynomial lower bounds for constant-depth algebraic circuits; and second, the functional lower bound framework of \cite{FSTW21} for size lower bounds on $\IPS$ proofs.
We adopt the same overall strategy, showing how the finite field setting introduces additional obstacles, and how we circumvent them.

Following \cite{GHT22}, we reduce the task of lower-bounding the size of a constant-depth algebraic circuit computing the multilinear polynomial that constitutes the $\IPS$ proof into the task of lower-bounding the size of a  constant-depth \emph{set-multilinear} \emph{circuit} computing a related \emph{set-multilinear polynomial}.
We derive this set-multilinear polynomial from the original multilinear $\IPS$ proof (which is not necessarily set-multilinear by itself) via the same variant of the functional lower bound method used in \cite{GHT22}.

\cite{GHT22} subsequently applies a reduction presented in \cite{LST21}, which converts constant-depth general circuits into constant-depth set-multilinear circuits.
Because the reduction presented in \cite{LST21} requires fields of sufficiently large characteristic, we rely on the recent extension of \citeallauthors{forbes24} \cite{forbes24}, which shows that this set-multilinearization reduction holds over all fields thereby removing this obstacle.
\cite{GHT22} also invokes another reduction presented in \cite{LST21} from a size lower bound of a set-multilinear formula into a rank lower bound of its coefficient matrix.
This second reduction already holds over all fields, and we use the improved parameters obtained by \citeallauthors{BDS24_journal}
\cite{BDS24_journal}.

The problem therefore reduces to constructing an unsatisfiable instance whose refutations, after the preceding reductions, have full rank.
\cite{GHT22} achieves this by introducing the knapsack polynomial, an instance that embeds a family of subset-sum instances ($\sum_{i=1}^n x_i - \beta = 0$, for $\beta > n$).
They then use the \textit{full} degree lower bound established in \cite{FSTW21} for refutations of subset-sum instances to obtain the required full rank lower bound.
Because the knapsack polynomial is satisfiable over Boolean assignments in finite fields, our task is to design a hard instance that both admits no satisfying Boolean assignment in finite fields and embeds a family of subset-sum-type instances that require full degree to refute. 
We proceed in two steps: first, we extend the full degree bound of \cite{FSTW21} to more general subset-sum-type instances; and second, we introduce a variant of the knapsack instance, knapsack mod $p$, that embeds a family of these more general subset-sum-type instances.
Thus we obtain the result. Although the theorem is stated over finite fields, it also holds over characteristic $0$ fields, thereby providing additional hard instances for the proof system studied on \cite{GHT22}.

As noted earlier, \cite{BLRS25} proves a lower bound for the same proof system as \Cref{thm:ips-modp-main-intro}, but under different field assumptions. \cite{BLRS25} requires size of the field to grow with the instance, whereas our lower bound holds for any field of constant positive characteristic. Consequently, in the constant positive characteristic setting, our result subsumes that of \cite{BLRS25} as it also holds over fixed finite fields. By contrast, the \cite{BLRS25} lower bound also covers characteristic $2$ and $3$ fields, which our result does not.

Both our work and \cite{BLRS25} also establish upper bounds for subsystems of constant-depth $\IPS$ over fields of positive characteristic.
While \cite{BLRS25} proves a general upper bound for systems stronger than constant-depth multilinear $\lbIPS$, the system for which we obtain a super-polynomial lower bound, our upper bound is for a single explicit instance within the same system.
Moreover, our specific instance is hard to refute in constant depth multilinear $\lbIPS$ over characteristic $0$ fields, where a corresponding lower bound was shown in \cite{GHT22}.
Hence we obtain the first separation between constant depth multilinear $\lbIPS$ over finite fields and the same system over characteristic $0$ fields.

Our separating instance $\ks_{w,e_2}$, the symmetric knapsack of degree $2$, is another variant of the knapsack instance used in \cite{GHT22}.
Note that the subset-sum instance can be viewed as an elementary symmetric sum of degree $1$.
In the same spirit as the knapsack polynomial, $\ks_{w,e_2}$ is designed so that it embeds a family of elementary symmetric sum of degree $2$.
\begin{theorem}[Informal; see \Cref{thm:separation}]
\label{thm:separation-intro}
    Let $p \geq 3$ be a prime, and let $\F$ be a field of characteristic $p$.
    Then, for the symmetric knapsack $\ks_{w,e_2}$ of degree $2$:
    \begin{itemize}
        \item $\ks_{w,e_2}$ has no satisfying Boolean assignment over $\F$, and over any field of characteristic $0$;
        \item there is a polynomial-size, constant-depth multilinear $\lbIPS$ refutation of $\ks_{w,e_2}$ over $\F$;
        \item for every characteristic $0$ field $E$, every constant-depth multilinear $\lbIPS$ refutation over $E$ of $\ks_{w,e_2}$ requires super-polynomial size.
    \end{itemize}
\end{theorem}

For our separation, we need an instance that has no Boolean satisfying assignment in either field, both the finite field and the characteristic $0$ field, rather than one whose satisfiability depends on the characteristic.
We establish this for $\ks_{w,e_2}$ by showing that elementary symmetric sums of certain degrees have no Boolean satisfying assignment in finite fields.
As $\ks_{w, e_2}$ admits no satisfying Boolean assignment in the finite field, it likewise admits none in the characteristic $0$ field.

The upper bound for refutations of $\ks_{w,e_2}$ in constant-depth multilinear $\lbIPS$ over finite fields follows from Fermat's little theorem.
What remains is the lower bound for $\ks_{w, e_2}$ over characteristic $0$ fields.
We show that, in characteristic $0$, every refutation of the elementary symmetric sum of degree $2$ must have full degree.
The corresponding $\IPS$ lower bound for $\ks_{w,e_2}$ then follows similarly to the argument in \cite{GHT22}.

\subsubsection{roABP-IPS Lower Bounds over Finite Fields}
We also present new lower bounds for \( \roAlbIPS \) over finite fields, using two distinct techniques: the Functional Lower Bound method and the Lower Bound by Multiples method. In both cases, we obtain finite-field analogues of results from~\cite{HLT24} and~\cite{FSTW21} respectively, which originally required fields of large characteristic. Moreover, our proofs are significantly simpler. As a first step, we establish an exponential lower bound for \( \roAlbIPS \) in any variable order.

\begin{theorem}[Informal; see \Cref{corollary: roABP lower bounds for any order}]
    Let $\F_q$ be a finite field with constant characteristic $q$. Then, there exists a polynomial $f \in \F[\wbar]$ such that any $\roAlbIPS$ refutation (in any variable order) of $f$ requires $2^{\Omega(n)}$-size.
\end{theorem}
This proof employs the Functional Lower Bound method and closely follows the strategy of~\cite{HLT24}. As in their work, we first establish a lower bound in a fixed variable order, and then extend the result to any order. However, our hard instance differs from theirs---this not only simplifies the argument, but also allows us to prove the result over fields of constant characteristic. Additionally, we provide a lower bound for an unsatisfiable system of equations.

\begin{theorem}[Informal; see \Cref{theorem: lower bounds by multiple}]
    Let $\F_q$ be a finite field of constant characteristic $q$. Then, there exist polynomials $f,g \in \F[x_1,\dots, x_n]$ such that the system of equations $f, g, \xbar^2- \xbar$ is unsatisfiable, and any $\roAlbIPS$ refutation (in any order of the variables) requires size $\exp\big({\Omega(n)}\big)$.
\end{theorem}
In this case, we apply the Lower Bound by Multiples method from~\cite{FSTW21}, and extend their result to finite fields. Our hard system of equations uses the same polynomial \( f \) as in their work, but a different choice of \( g \), which allows us to avoid their reliance on large characteristic fields. 

We emphasize that the lower bounds we obtain for \( \roAlbIPS \) are \emph{placeholder} lower bounds—that is, the hard instances considered are not efficiently computable by roABPs. This makes the model strictly weaker than the non-placeholder setting. In fact,  we show that it is \emph{impossible} to obtain non-placeholder lower bounds for \( \roAlbIPS \) over finite fields using the functional lower bound method (see \cite[Theorem 1 in full version]{HLT24} for a precise definition of ``the functional lower bound method'').

\begin{theorem}[\Cref{thm:flbm_limitation}]
    The functional lower bound method cannot establish non-placeholder lower bounds on the size of \( \roAlbIPS \) refutations when working in finite fields.
\end{theorem}

\subsubsection{Towards Hard CNF Formulas}
Several lower bounds are known for purely algebraic instances  against subsystems of $\IPS$. This raises an important question: Could we get lower bounds for CNF formulas against subsystems of $\IPS$ from those lower bounds? 

Note that an instance consisting of a set of polynomials written as \emph{circuit equations} $\left\{f_i(\bar{x})=0\right\}_i$, for $f_i(\bar{x}) \in \mathbb{F}[\bar{x}]$, does not necessarily correspond to a Boolean instance or a CNF formula. Specifically, we say that such an instance is \emph{Boolean} whenever $f_i(\bar{x}) \in\{0,1\}$ for $\bar{x} \in\{0,1\}^{|\bar{x}|}$. For example, a CNF written as a set of (polynomials representing) clauses is a Boolean instance. Similarly, the standard arithmetization of propositional formulas leads to Boolean instances. On the other hand, the instances used in \Cref{thm:ips-modp-main-intro} as well as the standard  subset sum $\sum_i x_i-\beta$ is non-Boolean, and thus said to be ``purely algebraic'': the image of the latter under $\{0,1\}$-assignments is $\{-\beta, 1-\beta, \ldots, n-\beta\}$, and thus cannot be considered a propositional or Boolean formula (formally, there is no known way to yield propositional proof lower bounds, say, in constant-depth Frege, from lower bounds for such purely algebraic instance even when such lower bounds are against proof systems that simulate constant-depth Frege). 

We solve this problem, and show how to attain propositional proof lower bounds from purely algebraic instances lower bounds. 
This is done using efficient bit-arithmetic in finite fields: from a circuit we derive the statements that express its gate-by-gate bit-arithmetic description.
we establish a \emph{translation lemma}---that is, we show that CNF encoding can be efficiently derived from circuit equations and vice versa within these subsystems of $\IPS$ in finite fields. If a subsystem of $\IPS$ can efficiently derive the CNF encoding and then refute it, a lower bound for circuit equations implies a lower bound for CNF formulas. 

In \cite{ST25}, Santhanam and Tzameret presented a translation lemma with \emph{extension axioms} in $\IPS$. In other words, given some additional axioms, $\IPS$ can efficiently derive the CNF encoding for circuit equations and vice versa. \textit{We eliminate the need to add additional extension axioms and extension variables altogether}: we show that without those additional axioms, already  \emph{bounded-depth} $\IPS$ over a finite field can efficiently derive the CNF encoding for bounded-depth circuit equations. Following our translation lemma, every superpolynomial lower bound for bounded-depth circuit equations against bounded-depth $\IPS$ implies a superpolynomial lower bound for CNF formulas against bounded-depth $\IPS$ over a finite field, and hence an $\ACZ[p]$-Frege lower bound following standard simulation of $\ACZ[p]$-Frege by constant-depth $\IPS$ over $\F_p$.

We now explain our translation lemma. 
\cite{ST25} used \emph{unary} encoding to encode $\CNF$s for circuit equations over finite fields. Each variable $x$ over a finite field $\F_q$ corresponds to $q$ bits $x_{q-1},\dots, x_0$ where $x_j$ equals $1$ for $0 \leq j \leq q- 1$  if and only if $x = j$; thus, these $q$ bits can be viewed as ``unary bits''.

We use the \textit{Lagrange polynomial}
\begin{equation*}
     \frac{\prod_{i=0, i \neq j}^{q-1}(x-i)}{\prod_{i=0, i \neq j}^{q-1}(j-i)}
\end{equation*}
to express each unary bit $x_j$ with variable $x$, which we call $\UBIT$
\begin{equation*}
    \UBIT_j(x) = \begin{cases}
        1, \quad x = j,\\
        0, \quad \text{otherwise}.
    \end{cases}
\end{equation*}

We introduce a notation called semi-$\CNF$s, which are $\CNF$s where each Boolean variable is substituted by the corresponding $\UBIT$. Hence, $\SCNF$s are substitution instances of $\CNF$s, which means a lower bound for $\SCNF$s implies a lower bound for $\CNF$s against sufficiently strong subsystems of $\IPS$, including bounded-depth $\IPS$.

We show that the semi-$\CNF$ encoding of all the extension axioms in \cite{ST25} can be efficiently proved in bounded-depth $\IPS$ over finite fields. Following the proof in \cite{ST25}, bounded-depth $\IPS$ can efficiently derive the semi-$\CNF$s encoding of circuit equations. Hence, a lower bound for circuit equations implies a lower bound for $\CNF$s.

\begin{theorem}[\Cref{corollary: lower bounds for circuit equation implies lower bounds for CNFs}]
        \label{theorem: lower bounds for circuit equation implies lower bounds for CNFs, intro}
        Let $\mathbb{F}_q$ be a finite field, and let $\{C(\overline{x})\}$ be a set of circuits of depth at most $\Delta$ in the Boolean variable $\overline{x}$. Then, if a set of circuit equations $\{C(\overline{x}) =0\}$ cannot be refuted in $S$-size, $O(\Delta^\prime)$-depth $\IPS$, then the CNF encoding of the set of circuit equations $\{\CNF(C(\overline{x})=0)\}$ cannot be refuted in $(S-\poly(|C|))$-size, $O(\Delta^\prime + \Delta)$-depth $\IPS$.
\end{theorem}

Notice that our lower bound \Cref{thm:ips-modp-main-intro} is against constant-depth $\IPS$ refutations which are \emph{multilinear}. Since our algebraic-to-CNF translation lemma requires non-multilinear proofs it is unclear how to carry the translation lemma for our hard instance in constant-depth multilinear $\IPS$. For this reason we cannot apply the translation lemma to our lower bound to obtain $\ACZ[p]$-Frege lower bounds. 

This aligns with the barrier discovered in \cite{HLT24}, in which proof systems closed under AND-introduction (i.e., from a set of formulas derive their conjunction), cannot use the Functional Lower Bound method (note that our lower bound in \Cref{thm:ips-modp-main-intro} employs this method). 

\medskip

Bit-arithmetic arguments are used in proof complexity in many works (beginning from \cite{Bus12}, and further in works such as \cite{AGHT20,IMP20,Gro23}, and as mentioned above \cite{ST25}). However, in all prior works the bit-arithmetic of a given circuit was not efficiently derived \emph{within the system} from the circuits themselves, rather it was used externally to argue about certain simulations. Thus, as far as we are aware of, our result is the first that shows how to efficiently derive internally within the proof system the bit-arithmetic from a circuit.

\section{Preliminaries}


\subsection{ Polynomials and Algebraic Circuits}\label{sec:algebraic_circuits}

For excellent treatises on algebraic circuits and their complexity see Shpilka and Yehudayoff \cite{SY10} as well as Saptharishi \cite{Sap17-survey}.
Let \G\ be a ring. Denote by $\G[X]$ the ring of (commutative) polynomials with coefficients from $ \G $ and variables $X:=\{x_1,x_2,\,\dots\,\}$. A \emph{polynomial }is a formal linear combination of monomials, where a \emph{monomial} is a product of variables. Two polynomials are \emph{identical }if all their monomials have the same coefficients. 

The (total) degree of a monomial is the sum of all the powers of variables in it. The (total) \emph{degree} of a polynomial is the maximal total degree of a monomial in it. The degree of an \emph{individual} variable in a monomial is its power. The \emph{individual degree} of a monomial is the maximal individual degree of  its variables. The individual degree of a polynomial is the maximal individual degree of its monomials. For a polynomial $f$ in $\G[X,Y]$ with $X,Y$ being pairwise disjoint sets of variables,  the \emph{individual $Y$-degree} of $f$ is the maximal individual degree of a $Y$-variable only in $f$.

Algebraic circuits and formulas over the ring \G\   compute polynomials in $\G[X]$ via addition and multiplication gates, starting from the input variables and constants from the ring. More precisely, an \emph{algebraic circuit} $C$ is a finite directed acyclic graph (DAG) with \textit{input nodes} (i.e., nodes  of in-degree zero) and a single \textit{output node} (i.e.,  a node of out-degree zero).  Edges are labelled by ring \G\ elements.  Input nodes are labelled with variables or scalars from the underlying ring. In this work (since we work with constant-depth circuits) all  other nodes have unbounded \emph{fan-in} (that is, unbounded in-degree) and are labelled by either an addition gate $+$ or a product gate $\times$.
Every node in an algebraic circuit $C$ \emph{computes} a polynomial in $\G[X]$ as follows: an input node computes  %
%
the variable or scalar that   labels  it. A $+$ gate
computes the linear combination of all the polynomials computed by its incoming nodes, where the coefficients of the linear combination are determined by the corresponding incoming edge labels. A $\times$ gate computes the product of all the polynomials computed by its incoming nodes (so edge labels in this case are not needed). The polynomial computed by a node $u$ in an algebraic circuit $C$ is denoted $\widehat u$. Given a circuit $C$, we denote by $\widehat C$ the polynomial computed by $C$, that is, the polynomial computed by the output node of $C$.  The \emph{\textbf{size}} of a circuit $C$ is the number of nodes in it, denoted $|C|$, and the \emph{\textbf{depth}} of a circuit is the length of the longest directed path in it (from an input node to the output node). The \textbf{\emph{product-depth }}of the circuit is the maximal number of product gates in a directed path from  an input node to the output node.


We say that a polynomial is \emph{homogeneous} whenever every monomial in it has the same (total) degree. We say that a polynomial is \emph{multilinear} whenever the individual degrees of each of  its variables are at most 1. 

Let $\xbar = \langle X_1,\ldots,X_d\rangle$ be a sequence of pairwise disjoint sets of variables, called \emph{a variable-partition}. We call a monomial $m$ in the variables $\bigcup_{i\in [d]}X_i$  \emph{set-multilinear} over the variable-partition $\xbar$ if it contains exactly one variable from each of the sets $X_i$, i.e. if there are $x_i\in X_i$  for all $i\in [d]$ such that $m = \prod_{i\in [d]}x_i$. A polynomial $f$ is set-multilinear over $\xbar$ if it is a linear combination of set-multilinear monomials over $\xbar$. For a sequence $\xbar$  of sets of variables, we denote by $\F_{\sml}[\xbar]$ the space of all polynomials that are set-multilinear over $\xbar$.

We say that an algebraic circuit $C$ is set-multilinear over $\xbar$ if $C$ computes a polynomial that is set-multilinear over $\xbar$, and each internal node of $C$  computes a polynomial that is set-multilinear over some sub-sequence of $\xbar$.

\subsubsection{Oblivious Algebraic Branching Programs}
An algebraic branching program (ABP) is a graph-based computational model for computing multivariate polynomials, providing a structured alternative to algebraic circuits. We state the formal definition below.

\begin{definition}[\cite{Nis91}; ABP]
Let $\mathbb{F}$ be a field. An \emph{algebraic branching program} (ABP) of \emph{depth}~$D$ and \emph{width} $\leq r$ over variables $x_1, \ldots, x_n$ is a directed acyclic graph (DAG) with the following properties:
\begin{enumerate}
    \item The vertex set is partitioned into $D+1$ layers $V_0, V_1, \ldots, V_D$, where $V_0$ contains a unique \emph{source} node $s$ and $V_D$ contains a unique \emph{sink} node $t$.
    \item All edges are directed from layer $V_{i-1}$ to $V_i$, for $1 \leq i \leq D$.
    \item Each layer satisfies $|V_i| \leq r$ for all $0 \leq i \leq D$.
    \item Each edge $e$ is labeled by a polynomial $f_e \in \mathbb{F}[x_1, \ldots, x_n]$.
\end{enumerate}
The \emph{(individual) degree} of the ABP is the maximum individual degree of any polynomial label $f_e$. The \emph{size} of the ABP is defined as $n \cdot r \cdot d \cdot D$, where $d$ denotes the (individual) degree. Each $s$–$t$ path computes a polynomial equal to the product of the edge labels along the path. The ABP as a whole computes the sum of these polynomials over all $s$–$t$ paths.

We define the following restricted variants of ABPs:
\begin{itemize}
    \item An ABP is called \emph{oblivious} if, for every layer $1 \leq \ell \leq D$, all edge labels between $V_{\ell - 1}$ and $V_\ell$ are univariate polynomials in a single variable $x_{i_\ell} \in \{x_1, \ldots, x_n\}$.
    \item An oblivious ABP is said to be a \emph{read-once oblivious ABP} (roABP) if each variable $x_i$ appears in the edge labels of exactly one layer. In this case, we have $D = n$, and the layers define a variable order, which we assume to be $x_1 < x_2 < \cdots < x_n$, unless otherwise stated.
    \item An oblivious ABP is said to be a \emph{read-$k$ oblivious ABP} if each variable $x_i$ appears in the edge labels of exactly $k$ layers, so that $D = kn$.
\end{itemize}
\end{definition}
We have the following fact about roABPs.

\begin{fact}
\label{fact:roABP}
    roABPs are closed under the following operations:
    \begin{itemize}
        \item If $f(\overline{x},\overline{y})\in\mathbb{F}$ is computable by a width-$r$ roABP in some variable order then the partial substitution $f(\overline{x}, \overline{\alpha})$, for $\alpha\in\mathbb{F}^|\overline{y}|$, is computable by a width-$r$ roABP in the induced order on $\overline{x}$, where the degree of this roABP is bounded by the degree of the roABP for $f$.
        \item If $f(z_1, \ldots , z_n)$ is computable by a width-$r$ roABP in variable order $z_1 < \ldots < z_n$, then $f(x_1y_1,\ldots, x_ny_n)$ is computable by a $\poly(r, \ideg f)$-width roABP in variable order $x_1 < y_1 <\ldots < x_n < y_n$.
    \end{itemize}
\end{fact}

\subsection{Strong Algebraic Proof Systems}
For a survey about algebraic proof systems and their relations to algebraic complexity see the survey \cite{PT16}.
Grochow and Pitassi~\cite{GP14}  suggested the following algebraic proof system  which is essentially a Nullstellensatz proof system \cite{BeameIKPP96} written as an algebraic circuit. A proof in the  Ideal Proof System is given as  a \emph{single} polynomial. We provide below the \emph{Boolean} version of  $\IPS$ (which includes the Boolean axioms), namely the version that establishes the unsatisfiability over 0-1 of a set of polynomial equations.  In what follows we follow the notation in \cite{FSTW21}:



\begin{definition}[Ideal Proof System ($\IPS$),
        Grochow-Pitassi~\cite{GP14}]\label{def:IPS} Let $f_1(\vx),\ldots,f_m(\vx),p(\vx)$ be a collection of polynomials in $\F[x_1,\ldots,x_n]$ over the field \F. An \demph{$\IPS$ proof of $p(\vx)=0$ from axioms $\{f_j(\vx)=0\}_{j\in [m]}$}, showing that $p(\vx)=0$ is semantically  implied from the assumptions $\{f_j(\vx)=0\}_{j\in [m]}$ over $0$-$1$ assignments, is an algebraic circuit $C(\vx,\vy,\vz)\in\F[\vx,y_1,\ldots,y_m,z_1,\ldots,z_n]$ such that (the equalities in what follows stand for  formal polynomial identities\footnote{That is, $C(\vx,\vnz,\vnz)$ computes the zero polynomial and $C(\vx,f_1(\vx),\ldots,f_m(\vx),x_1^2-x_1,\ldots,x_n^2-x_n)$ computes the polynomial $p(\vx)$.}; recall the notation $\widehat C$ for the \emph{polynomial} computed by circuit $C$):
        \begin{enumerate}
                \item $\widehat C(\vx,\vnz,\vnz) = 0$;
                \item $\widehat C(\vx,f_1(\vx),\ldots,f_m(\vx),x_1^2-x_1,\ldots,x_n^2-x_n)=p(\vx)$.
        \end{enumerate}
        The \demph{size of the $\IPS$ proof} is the size of the circuit $C$. An \IPS\ proof  $C(\vx,\vy,\vz)$ of  $1=0$ from $\{f_j(\vx)=0\}_{j\in[m]}$ is called an \demph{$\IPS$ refutation} of $\{f_j(\vx)=0\}_{j\in[m]}$ (note that in this case  it must hold that  $\{f_j(\vx)=0\}_{j\in [m]}$ have no common solutions in $\bits^n$).
        If $\widehat C$ is of individual degree $\le 1$ in each $y_j$ and $z_i$, then this is a \demph{linear} $\IPS$ refutation (called \emph{Hilbert} $\IPS$ by Grochow-Pitassi~\cite{GP14}), which we will abbreviate as \lIPS. If $\widehat C$ is of individual degree $\le 1$ only in the $y_j$'s then we say this is an \lbIPS\ refutation (following \cite{FSTW21}). If  $\widehat C(\vx,\vy,\vnz)$ is of individual degree $\le 1$ in each $x_j$ and $y_i$ variables, while $\widehat C(\vx,\vnz,\vz)$ is not necessarily multilinear, then this is a \demph{multilinear} \lbIPS\ refutation. 
        
        If $C$ is of depth at most $d$, then this is  called a depth-$d$ \IPS\ refutation, and further called a depth-$d$ \lIPS\ refutation if $\widehat C$ is linear in $\vy,\vz$, and a depth-$d$ \lbIPS\ refutation if $\widehat C$ is linear in $\vy$, and depth-$d$ multilinear \lbIPS\ refutation if $\widehat C(\vx,\vy,\vnz)$ is linear in $\vx,\vy$. 
\end{definition}

Notice that the definition above adds the  equations $\{x_i^2-x_i=0\}_{i=1}^n$, called the  \demph{Boolean axioms} denoted  $\vx^2-\vx$, to the system $\{f_j(\vx)=0\}_{j=1}^m$. This allows  to refute over $\bits^n$ unsatisfiable systems of equations. The variables $\vy,\vz$ are  called the \emph{placeholder} \emph{variables} since they are used as placeholders for the axioms. Also, note that the first equality in the definition of $\IPS$ means that the polynomial computed by $C$ is in the ideal generated by $\overline y,\overline z$, which in turn, following the second equality, means that $C$ witnesses the fact that $1$ is in the ideal generated by $f_1(\vx),\ldots,f_m(\vx),x_1^2-x_1,\ldots,x_n^2-x_n$ (the existence of this witness, for unsatisfiable set of polynomials, stems from the Nullstellensatz \cite{BeameIKPP96}).    
\medskip 

In this work we focus on multilinear \lbIPS\ refutations. This proof system is complete because its \emph{weaker} subsystem multilinear-formula \lbIPS\  was shown in \cite[Corollary 4.12]{FSTW21} to be complete (and to simulate Nullstellensatz with respect to sparsity by already depth-2 multilinear \lbIPS\ proofs). 

To build an intuition for multilinear \lbIPS\ it is useful to consider a subsystem of it in which refutations are written as 
$$
C(\vx,\vy,\vz)= \sum_i g_i(\vx)\cdot y_i + C'(\vx,\vz),$$ where $\widehat C'(\vx,\vnz)=0$ and the $g_i$'s are multilinear. Note indeed that $C(\vx,\vnz,\vnz)=0$ so that the first condition of $\IPS$ proofs holds, and that $C(\vx,\vy,\vnz)$ is indeed multilinear in $\vx,\vy$.


\medskip 

\textbf{Important remark}: Unlike the multilinear-formula \lbIPS\ in \cite{FSTW21}, in  multilinear \lbIPS\ refutations $C(\vx,\vy,\vz)$ we do \emph{not} require that the refutations are written as multilinear \emph{formulas} or multilinear \emph{circuits}, only that the \emph{polynomial} \emph{computed }by $C(\vx,\vy,\vnz)$ is multilinear, hence the latter proof system easily simulates the former. 

\medskip

We now formally state how we prove a functional lower bound for $\mathcal{C}$-$\IPS$ systems.

\begin{theorem}[Functional Lower Bound Method; Lemma 5.2 in {\cite{FSTW21}}]
\label{thm:func_lb_method}
Let \( \mathcal{C} \subseteq \mathbb{F}[\overline{x}] \) be a circuit class, and let \( f(\overline{x}) \in \mathcal{C} \) be a polynomial, which has no boolean roots. A \emph{functional lower bound against \( \mathcal{C}\text{-}\lbIPS \)} for \( f(\overline{x}) \) and \( \overline{x}^2 - \overline{x} \) is a lower bound argument using the following circuit lower bound for \( \frac{1}{f(\overline{x})} \): Suppose that \( g \notin \mathcal{C} \) for all \( g \in \mathbb{F}[\overline{x}] \) with
\[
g(\overline{x}) = \frac{1}{f(\overline{x})}, \qquad \forall \, \overline{x} \in \{0,1\}^n. \tag{1.1}
\]
Then, \( f(\overline{x}) \) and \( \overline{x}^2 - \overline{x} \) do not have \( \mathcal{C}\text{-}\lbIPS \) refutations. Moreover, if \( \mathcal{C} \) is a set of multilinear polynomials, then \( f(\overline{x}) \) and \( \overline{x}^2 - \overline{x} \) do not have \( \mathcal{C}\text{-}\IPS \) refutations.
\end{theorem}

\subsection{Coefficient Matrix and Dimension}
We define notions and measures used in this paper. Consider a polynomial \( f \in \mathbb{F}[\overline{x}, \overline{y}] \). We can construct this polynomial by organizing the coefficients of \( f \) into a matrix format: the rows are indexed by monomials \( \overline{x}^{\overline{a}} \) in the \( \overline{x} \)-variables, the columns are indexed by monomials \( \overline{y}^{\overline{b}} \) in the \( \overline{y} \)-variables, and the entry at position \((\overline{x}^{\overline{a}}, \overline{y}^{\overline{b}})\) is the coefficient of the monomial \( \overline{x}^{\overline{a}} \overline{y}^{\overline{b}} \) in \( f \).

\begin{definition}[Coefficient Matrix]
\label{def:coefficient_matrix}
Let $f \in \mathbb{F}[\overline{x}, \overline{y}]$ be a polynomial, where $\overline{x} = \{x_1, \ldots, x_n\}$ and $\overline{y} = \{y_1, \ldots, y_m\}$. Let $\coeff_{\overline{x}^{\overline{a}}\overline{y}^{\overline{b}}}(f)$ denote the coefficient of the monomial $\overline{x}^{\overline{a}} \overline{y}^{\overline{b}}$ in $f$. The \emph{coefficient matrix} of $f$ is the matrix $C_f$ with entries
\[
(C_f)_{\overline{a}, \overline{b}} := \coeff_{\overline{x}^{\overline{a}} \overline{y}^{\overline{b}}}(f),
\]
such that $\sum_{i=1}^n a_i + \sum_{j=1}^m b_j \leq \deg(f)$.
\end{definition}

For our purposes, we care about the dimension of this matrix.
\begin{definition}[Coefficient space]
\label{def:coefficient_space}
    Let $\coeff_{\overline{x}|\overline{y}}: \mathbb{F}[\overline{x}, \overline{y}] \rightarrow 2^{\mathbb{F}[\overline{x}]}$ be the space of $\mathbb{F}[\overline{x}][\overline{y}]$ coefficients, defined by
    $$
    \coeffs{{\overline{x}|\overline{y}}}(f) := \Big \{ \coeff_{\overline{x}|\overline{y}^{\overline{b}}} (f) \Big \}_{\overline{b}\in\mathbb{N}^n},
    $$
where the coefficients of $f$ are in $\mathbb{F}[\overline{x}][\overline{y}]$. Similarly we have $\coeffs{{\overline{y}|\overline{x}}}(f)$ by taking coefficients in $\mathbb{F}[\overline{y}][\overline{x}]$
\end{definition}
That is, we use the above in the context of \emph{coefficient dimension}, where we look at the dimension of the coefficient space of $f$, denoted $\dim \coeffs{{\overline{x}|\overline{y}}}(f)$. We state a result that connects this to the rank of the matrix.

\begin{lemma}[Coefficient matrix rank equals dimension of polynomial space; Nisan \cite{Nis91}]
    Consider $f\in \mathbb{F}[\overline{x}, \overline{y}]$, and let $C_f$ denote the coefficient matrix of $f$ (\Cref{def:coefficient_matrix}). Then, the following holds:
    $$
    \rank C_f = \dim \coeffs{{\overline{x}|\overline{y}}}(f) = \dim \coeffs{{\overline{y}|\overline{x}}}(f).
    $$
\end{lemma} 

We now show that the coefficient dimension in fact characterizes the width of roABPs.
\begin{lemma}[roABP width equals coefficient dimension]
\label{lem:roABP_width_equals_coeff_dim}
Let \( f \in \mathbb{F}[\overline{x}] \) be a polynomial. If \( f \) is computed by a roABP of width \( r \), then
\[
r \geq \max_i \dim \coeffs{{\overline{x}_{\leq i} \mid \overline{x}_{> i}}}(f).
\]
Conversely, \( f \) can be computed by a roABP of width $\max_i \dim \coeffs{{\overline{x}_{\leq i} \mid \overline{x}_{> i}}}(f)$.
\end{lemma}

The coefficient dimension of a polynomial \( f(\overline{x}, \overline{y}) \) measures its complexity by considering the span of all coefficient vectors with respect to the \( \overline{y} \)-monomials. In a similar vein, we also consider the \emph{evaluation dimension}, introduced by Saptharishi \cite{Saptharishi12}. Specifically, it captures the dimension of the span of all evaluations of \( f(\overline{x}, \cdot) \) at points \( \overline{y} \in \mathbb{F}^m \). 

\begin{definition}[Evaluation dimension] 
\label{def:eval}
    Let $S\subseteq \mathbb{F}$. Let $\mathbf{Eval}_{\overline{x}|\overline{y}, S}: \mathbb{F}[\overline{x},\overline{y}] \rightarrow 2^{\mathbb{F}[\overline{x}]}$ be the space of $\mathbb{F}[\overline{x},\overline{y}]$ evaluations over $S$, defined by
    $$
    \mathbf{Eval}_{\overline{x}|\overline{y}, S}(f(\overline{x},\overline{y})) := \{f(\overline{x},\overline{\beta})\}_{\overline{\beta} \in S^{|\overline{y}|}}.
    $$
    The \emph{evaluation dimension} is therefore the dimension of the above space, denoted $\dim \evals{\overline{x} \mid \overline{y}, S}(f)$.
\end{definition}
That is, we consider the span of functions $f$ over all assignments to the $\overline{y}$ variables. This measure is particularly useful for our applications, as it is directly related to the coefficient dimension.

\begin{lemma}[Evaluation dimension bounds coefficient dimension; Forbes-Shpilka~\cite{ForbesShpilka13b}]
\label{lem:eval_dim_lb_coeff_dim}
Let \( f \in \mathbb{F}[\overline{x}, \overline{y}] \), and let \( S \subseteq \mathbb{F} \). Then,
\[
\evals{\overline{x} \mid \overline{y}, S}(f) \subseteq \spant\coeffs{_{\overline{x} \mid \overline{y}}}(f),
\]
and hence,
\[
\dim \evals{\overline{x} \mid \overline{y}, S}(f) \leq \dim \coeffs{{\overline{x} \mid \overline{y}}}(f).
\]
Moreover, if \( |S| > \mathrm{ideg}(f) \), then equality holds:
\[
\dim \evals{\overline{x} \mid \overline{y}, S}(f) = \dim \coeffs{{\overline{x} \mid \overline{y}}}(f).
\]
\end{lemma}

%
% ...

\subsection{Set-Multilinear Monomials over a Word}\label{sec:prelim:setmulti}

We recall some notation from \cite{LST21}. Let $w\in\ZZ^d$ be a word. For a subset $S\subseteq [d]$ denote by $w_S$ the sum $\sum_{i\in S} w_i$, and by $w|_S$ the \textbf{subword} of $w$ indexed by the set $S$. Let\footnote{The $P_w$ here is not to be confused with the canonical full-rank set-multilinear polynomial in \cite{LST21} denoted as well by $P_w$ mentioned in the introduction.}
\[
P_w := \{i\in [d] : w_i\geq 0\}
\]
be the set of \textbf{positive indices} of $w$ and let
\[
N_w := \{i\in [d] : w_i < 0 \}
\]
be the set of \textbf{negative indices} of $w$.

Given a word $w$, we associate with it a sequence $\overline{X}(w) = \langle X(w_1),\ldots,X(w_d)\rangle$ of sets of variables, where for each $i\in [d]$ the size of $X(w_i)$ is $2^{|w_i|}$. We call a monomial set-multilinear over a word $w$ if it is set-multilinear over the sequence $\xbar(w)$.

For a word $w$, let $\Pi_w$ denote the projection onto the space $\F_\sml[\xbar(w)]$, which maps set-multilinear monomials over $w$ identically to themselves and all other monomials to $0$.
When the underlying variable partition is clear from context, we simply write $\Pi_{\textnormal{sml}}$ to denote the set-multilinear projection.

\subsection{Relative Rank}

Let $M^P_w$ and $M^N_w$ denote the set-multilinear monomials over $w|_{P_w}$ and $w|_{N_w}$, respectively. Let $f\in\F_{\sml}[\xbar(w)]$ and denote by $M_w(f)$ the matrix with rows indexed by $M^P_w$ and columns indexed by $M^N_w$, whose $(m,m')$-th entry is the coefficient of the monomial $mm'$ in $f$.

For any $f\in\F_{\sml}[\xbar(w)]$ define the \textbf{relative rank} with respect to $w$ as follows
\[
\relrk_w(f) = \frac{\rank(M_w(f))}{\sqrt{|M^P_w|\cdot|M^N_w|}}.
\]

\subsection{Monomial Orders}

Finally we recall some basic notions related to monomial orders. For an in-depth introduction see \cite{CoxLittleOShea15}. A monomial order (in a polynomial ring $\F[X]$) is a well-order $\leq$ on the set of all monomials that respects multiplication:
\[
\text{if }m_1\leq m_2\text{, then } m_1m_3\leq m_2m_3\text{ for any }m_3.
\]
It is not hard to see that any monomial order extends the submonomial relation: if $m_1m_2 = m_3$  for some monomials $m_1,m_2$ and $m_3$, then $m_1\leq m_3$. This is essentially the only property we need of monomial orderings, and thus our results work for any monomial ordering. Given a polynomial $f\in\F[X]$, the leading monomial of $f$, denoted $\LM(f)$, is the highest monomial with respect to $\leq$ that appears in $f$  with a non-zero coefficient. We conclude this section with the following known fact.

\begin{lemma}
\label{lem:span}
    For any set of polynomials $S\subseteq \mathbb{F}[\overline{x}]$, the dimension of their span in $\mathbb{F}[\overline{x}]$ is equal to the number of unique distinct leading or trailing monomials in their span:
    $$
    \dim \spant S = | \text{LM}(\spant S)| = | \text{TM}(\spant S)|,
    $$
    where LM and TM stand for leading and trailing monomials respectively. In particular, we have
    $$
     \dim \spant S \geq  | \text{LM}(S)|, | \text{TM}(S)|.
    $$
\end{lemma}

\section{Lower Bounds for Constant-depth Multilinear IPS}

\label{sec:cd-IPS}
    \subsection{Notation for Knapsack}
    Before defining our hard instance, we introduce some notation.
    Our construction is based on the instance $\ks_w$ from \cite{GHT22}, and we adopt parts of their notation.
    
    Let $w \in \ZZ^d$ be an arbitrary word. 
    Consider the sequence $\overline{X}(w)=\langle X(w_1),\dots,X(w_d)\rangle $ of sets of variables and the following useful representation of the variables in $\overline{X}(w)$.
    For any $i\in P_w$, we write the variables of $X(w_i)$ in the form $x^{(i)}_\sigma$, where $\sigma$ is a binary string indexed by the set (formally, a binary string  \emph{indexed} by a set $A$ is a function from $A$ to $\{0,1\}$):   
\[
A_w^{(i)} := \left[\sum_{\substack{i'\in P_w\\i' < i }}w_{i'} + 1,\sum_{\substack{i'\in P_w\\i'\leq i}}w_{i'}\right].
\] 
Hence, the size of $A_w^{(i)}$ is precisely $w_i$, which implies that there are $2^{|A_w^{(i)}|}=2^{w_i}$ possible strings indexed by $A_w^{(i)}$, each corresponding to a distinct variable in $X(w_i)$.

Similarly, for any $j\in N_w$, we  write the variables of $X(w_i)$ in the form $y^{(j)}_\sigma$, where $\sigma$ is a binary string indexed by the set 
\[
B_w^{(j)} := \left[\sum_{\substack{j'\in N_w\\ j' < j}}|w_{j'}| + 1,\sum_{\substack{j'\in N_w\\ j' \leq j}}|w_{j'}|\right].
\]
We call the variables in $x^{(i)}_\sigma$ the \emph{positive variables}, or simply $\vx$-variables, and the variables $y^{(j)}_\sigma$ the \emph{negative variables}, or simply $\vy$-variables. We write $A_w^S$ for the set $\bigcup_{i\in S}A_w^{(i)}$ for any $S\subseteq P_w$, and $B_w^T$ for the set $\bigcup_{j\in T}B_w^{(j)}$ for any $T\subseteq N_w$.

Each monomial that is set-multilinear on $w|_{S}$ for some $S\subseteq P_w$ corresponds to a binary string indexed by the set $A_w^S$.
Similarly, each monomial that is set-multilinear on $w|_{T}$ for some $T\subseteq N_w$ corresponds to a binary string indexed by the set $B_w^T$.
For any set-multilinear monomial $m$ on some $w|_{S}$ with $S\subseteq P_w$, we denote by $\sigma(m)$ the corresponding binary string indexed by $A_w^S$.
Conversely, for any binary string $\sigma$ indexed by $A_w^S$, we denote by $m(\sigma)$ the monomial it defines.
The same correspondence holds for strings and monomials on the negative variables.
Thus, observe that for any (negative or positive) monomial $m$, we have  $m(\sigma(m))=m$.
Moreover, if $m$ is a negative monomial and $S\subseteq P_w$, we write $m(\sigma(m)|_{A^S_w})$ to denote the \emph{positive} monomial determined by the string $\sigma(m)|_{A^S_w}$, which is a substring of $\sigma(m)$ restricted to $A^S_w$.

Therefore, every set-multilinear monomial on $w$ has degree $d$, with each $\vx$-variable picked uniquely from the $X(w_i)$-variables for $i\in P_w$ (the positive indices in $w$), and each $\vy$-variable picked uniquely from the $X(w_j)$-variables for $j\in N_w$ (the negative indices).
Moreover, such a set-multilinear monomial on $w$ corresponds to a binary string of length $\sum_{i=1}^d |w_i|$.

We define the \textbf{overlap graph} $G$ of the word $w$ as the bipartite graph $(P_w, N_w, E)$, with an edge between $i \in P_w$ and $j\in N_w$ if $A_w^{(i)}$ and $B_w^{(j)}$ overlap, that is
\[
E = \{(i,j) \mid A_w^{(i)}\cap B_w^{(j)}\neq \emptyset\}.
\]
We say that the word $w$ is \textbf{balanced} if for every $i \in P_w \cup N_w$, the neighbourhood $N_{G}(i)$ is non-empty (see \Cref{fig:balanced}).
In what follows, we suppose that $|w_{N_w}|\geq |w_{P_w}|$, so that the negative monomials are determined by longer binary strings than the positive ones. Otherwise, we flip the roles of the positive and negative variables in the definition below.
We define the \textbf{positive overlap} of $w$, denoted $\Delta_G(P_w)$, as the maximum degree of a vertex in $P_w$.
Similarly, the \textbf{negative overlap}, denoted $\Delta_G(N_w)$, is the maximum degree of a vertex in $N_w$.
A partition of the set of positive indices $P_w = P_w^{(1)} \sqcup \cdots\sqcup P_w^{(r)}$ is called \textbf{scattered} if for every part, the neighbourhoods of positive indices in that part are pairwise disjoint, that is
\[N_{G}(i_1) \cap N_{G}(i_2) = \emptyset \quad
(\forall j \in [r], i_1, i_2 \in P_w^{(j)}, i_1 \neq i_2).
\]

\begin{figure}
    \centering
    \includegraphics[scale=0.7]{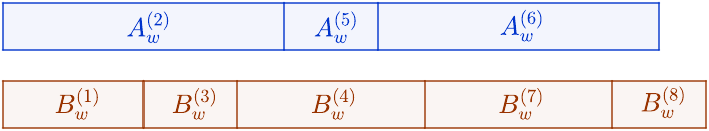}
    \caption{\small From \cite{GHT22}. Illustration of a word $w$. Each index $w_i$ of $w$ is shown as a box with $w_i$ slots, so every variable $x^{(i)}_\sigma$ in $X(w_i)$ appears as the string $\sigma$ written inside its corresponding box. The word $w$ shown is balanced.}
    \label{fig:balanced}
\end{figure}

\subsection{Hard Instance: Knapsack mod $p$}

 We now construct our hard instance $\ks_{w,p}$, knapsack mod $p$.
 Let $w \in \ZZ^d$ be a word with $|w_i| \leq b$ for every $i$.  
For $i \in P_w$ and $\sigma \in \{0,1\}^{A_w^{(i)}}$ let 
    \begin{equation*}
        f_\sigma^{(i)} := \prod_{\substack{j \in N_w \\ A_w^{(i)} \cap B_w^{(j)} \neq \emptyset}} f_\sigma^{(i,j)},
    \end{equation*}
where
\begin{equation}
    \label{eqn:knapsack-f-sigma-ij}
    f_\sigma^{(i,j)} := 1 - \prod_{\sigma_j \in \{0,1\}^{B_w^{(j)}}}\left(1-y_{\sigma_j}^{(j)} \right),
\end{equation}
where the product in \Cref{eqn:knapsack-f-sigma-ij} ranges over all those $\sigma_j$ that agree with $\sigma$ on $A_w^{(i)} \cap B_w^{(j)}$ (see \Cref{fig:matching}).
Let $P_w = P_w^{(1)} \sqcup \cdots\sqcup P_w^{(r)}$ be a scattered partition of the set of positive indices such that $r < p$.
We define our hard instance
\[
\ks_{w,p} := \sum_{j \in [r]} \prod_{i \in P_w^{(j)}} \ks_{w,p}^{(i)}  - \beta,
\]
where
\[
\ks_{w,p}^{(i)} := 1 - \ml  \left(\sum_{\sigma \in \{0,1\}^{A_w^{(i)}}} x_\sigma^{(i)}f_\sigma^{(i)} \right)^{p-1},
\]
and $\beta \in \FF$ is chosen such that $\ks_{w,p}$ is unsatisfiable over Boolean assignments.

\textbf{Comment} (the existence of $\beta$):
We observe that each $f_\sigma^{(i)}$ is a Boolean function. Hence, by Fermat's little theorem, each $\ks_{w,p}^{(i)}$ is a Boolean function.
It follows that
\[
\sum_{j \in [r]} \prod_{i \in P_w^{(j)}} \ks_{w,p}^{(i)} \in \{0,1, \ldots,r\}.
\]
Since $r<p$, we can always choose $\beta \in \FF$ such that $\ks_{w,p}$ is unsatisfiable over Boolean assignments.

\textbf{Comment} (computing $\ks_{w,p}$ by a $\poly(d,2^{bp})$-size, product-depth $3$, multilinear formula of degree $O(pdb2^b)$):
Fix $i \in P_w$ and consider computing $\ks_{w,p}^{(i)}$.
Let $\sigma_1, \sigma_2 \in \{0,1\}^{A_w^{(i)}}$ be distinct strings.
Suppose there exists $j \in N_w$ such that $A_w^{(i)} \cap B_w^{(j)} \neq \emptyset$ and the polynomials $f_{\sigma_1}^{(i,j)}$ and $f_{\sigma_2}^{(i,j)}$ share $\ybar$-variables. Then, by construction, $f_{\sigma_1}^{(i,j)} = f_{\sigma_2}^{(i,j)}$.
It follows that $\ml(f_{\sigma_1}^{(i)}f_{\sigma_2}^{(i)})$ can be computed in the same way as $f_{\sigma_1}^{(i)}f_{\sigma_2}^{(i)}$, but with the shared $f_{\sigma_2}^{(i,j)}$ terms excluded from the construction of $f_{\sigma_2}^{(i)}$.
Hence, $\ks_{w,p}^{(i)}$ can be computed by a product-depth $2$, multilinear formula of size $\poly(2^{bp})$.
Since $P_w = P_w^{(1)} \sqcup \cdots\sqcup P_w^{(r)}$ is a scattered partition of the positive indices, the variables in each $\ks_{w,p}^{(i)}$ are disjoint across distinct $i$.
Therefore, $\ks_{w,p}$ can be computed by a product-depth $3$, multilinear formula of size $\poly(d,2^{bp})$.
Moreover, each $f_\sigma^{(i)}$ has degree at most $O(b2^b)$, so $\ks_{w,p}^{(i)}$ has degree at most $O(pb2^b)$.
The overall degree of $\ks_{w,p}$ is therefore $O(pdb2^b)$.

\begin{figure}
    \centering
    \includegraphics[scale=0.7]{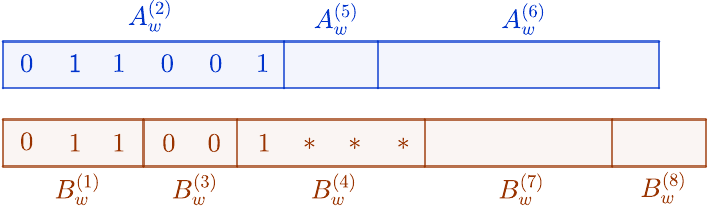}
    \caption{\small From \cite{GHT22}. Here $\ast$ represents either $0$ or $1$. In the construction of the polynomial $\ks_{w,p}$, for $i=2$ and $\sigma = 011001$, we see that $f_{011001}^{(2)} =  y_{011}^{(1)} \cdot y_{00}^{(3)} \cdot (1- (1-y_{1000}^{(4)}) (1-y_{1001}^{(4)}) \cdots (1-y_{1111}^{(4)} ))$. While our construction of $f_\sigma^{(i)}$ differs from \cite{GHT22}, it still functions as an indicator for the variable $x_\sigma^{(i)}$.}
    \label{fig:matching}
\end{figure}

\subsection{Degree Lower Bound}

We now state and prove the degree lower bound that we use in the rank lower bound.
We begin with the bound that was used in \cite{GHT22}.

\begin{lemma}[\cite{FSTW21} Proposition 5.3]
\label{subsetsumdegree}
    Let $n \geq 1$, $\chara(\FF) = 0$ or $\chara(\FF) > n$, and $\beta \in \mathbb{F} \setminus \{0,1,\ldots,n\}$.
    If $f \in \mathbb{F}[x_1, \ldots, x_n]$ is the multilinear polynomial such that
    \[
    f(\xbar) \left(\sum_{i \in [n]} x_i - \beta \right) = 1 \pmod {\xbar^2 - \xbar},
    \]
    then $\deg f = n$.
\end{lemma}

\begin{lemma}
\label{lem:deg-bound-general}
    Let $\xbar = \bigsqcup_{i \in I} \xbar_i$ be a partition of the variables $\xbar = \{x_1, \ldots, x_n\}$, $\chara (\FF)= 0$ or $\chara(\FF) = p > |I|$, and $\beta \in \mathbb{F} \setminus \{0,1,\ldots,|I|\}$.
    For $i \in I$, let $\psi_i \in \FF[\xbar_i]$ be a polynomial over the $\xbar_i$-variables that is multilinear, full degree (that is, $\deg \psi_i = |\xbar_i|$) and a Boolean function.
    If $f \in \FF[\xbar]$ is the multilinear polynomial such that
    \begin{equation}
        f(\xbar) \left(\sum_{i \in I} \psi_i - \beta \right) = 1 \pmod {\xbar^2 - \xbar},
    \end{equation}
    then $\deg f = n$.
\end{lemma}
\begin{proof}
    Let $\{w_i\}_{i \in I}$ be Boolean variables and $f_w \in \FF[\wbar]$ be the multilinear polynomial such that
    \begin{equation}
    \label{eq:deg-bound-general-aux}
        f_w(w_1, \ldots, w_{|I|}) \left( \sum_{i \in I} w_i - \beta \right) = 1 \pmod {\wbar^2 - \wbar}.
    \end{equation}
    By \Cref{subsetsumdegree}, we have $\deg f_w = |I|$.
    We show that the following polynomial identity over the $\xbar$-variables holds:
    \begin{equation}
    \label{eq:deg-bound-general-sub}
        f_w(\psi_1, \ldots, \psi_{|I|}) \left( \sum_{i \in I} \psi_i - \beta \right) = 1 \pmod {\xbar^2 - \xbar}.
    \end{equation}
    We note that $f_w$ is a polynomial over the $\wbar$-variables that is multilinear and for every $i \in I$, $\psi_i$ is a polynomial over the $\xbar_i$-variables that is multilinear.
    Therefore, as $\xbar = \bigsqcup_{i \in I} \xbar_i$ is a partition of the $\xbar$-variables, $f_w(\psi_1, \ldots, \psi_{|I|})$ is a polynomial over the $\xbar$-variables that is multilinear.
    Thus, to show \Cref{eq:deg-bound-general-sub}, it suffices to show that
    \begin{equation}
    \label{eq:deg-bound-general-bool}
        f_w(\psi_1, \ldots, \psi_{|I|}) \left( \sum_{i \in I} \psi_i - \beta \right) = 1
    \end{equation}
    holds for all $\xbar \in \{0,1\}^n$.
    Let $\alpha \in \{0,1\}^n$ be a Boolean assignment for the $\xbar$-variables and, for all $i \in I$, let $w_i = \psi_i(\alpha|_{\xbar_i})$.
    Since, for all $i \in I$, $\psi_i$ is a Boolean function, the assignments on the $\wbar$-variables are all Boolean assignments. 
    Therefore, for these Boolean assignments on the $\wbar$-variables, by \Cref{eq:deg-bound-general-aux},
    \[
         f_w(w_1, \ldots, w_{|I|}) \left( \sum_{i \in I} w_i - \beta \right) = 1.
    \]
    Since $w_i = \psi_i(\alpha|_{\xbar_i})$ for all $i \in I$, we see that \Cref{eq:deg-bound-general-bool} holds for the Boolean assignment $\alpha$ on the $\xbar$-variables.
    We therefore see that \eqref{eq:deg-bound-general-sub} holds.
    Thus, $f = f_w(\psi_1, \ldots, \psi_{|I|})$.
    Finally, as $f_w$ has full degree and for all $i \in I$, $\psi_i$ has full degree, we see that $f_w(\psi_1, \ldots, \psi_{|I|})$ has full degree.
    Therefore, $\deg f = \deg f_w(\psi_1, \ldots, \psi_{|I|}) = n$.
\end{proof}
    
\begin{corollary}
\label{lem:deg-bound-special}
     Let $\xbar = \bigsqcup_{i \in I} \xbar_i$ be a partition of the variables $\xbar = \{x_1, \ldots, x_n\}$, $\chara (\FF)= 0$ or $\chara(\FF) = p > |I|$, and $\beta \in \mathbb{F} \setminus \{0,1,\ldots,|I|\}$.
    If $f \in \FF[\xbar]$ is the multilinear polynomial such that
    \begin{equation}
    \label{eqnpartitiondegree}
        f(\xbar) \left(\sum_{i \in I} \prod_{x \in \xbar_i} (1-x) - \beta \right) = 1 \pmod {\xbar^2 - \xbar},
    \end{equation}
    then $deg (f) = n$.
\end{corollary}
\begin{proof}
    This follows from \Cref{lem:deg-bound-general}, taking for all $i \in I$,
    \[
        \psi_i(\xbar_i) = \prod_{x \in \xbar_i} (1-x),
    \]  
    and noting that $\psi_i$ is a multilinear, full degree polynomial and a Boolean function.
\end{proof}

\subsection{Rank Lower Bound}

\begin{lemma}
\label{lem:rank-lower-bound}
    Let $\mathbb{F}$ be a field with characteristic $p$ and $w \in \mathbb{Z}^d$ be a balanced word.
    If $f$ is the multilinear polynomial such that
    \[
    f = \frac{1}{\ks_{w,p}} \text{ over Boolean assignments,}
    \]
    then $M_w(f)$ has full rank.
\end{lemma}
\begin{proof}
    We recall the assumption that $|w_{N_w}| \geq |w_{P_w}|$ from the construction of $\ks_{w,p}$.
    Now write
    \begin{equation}
    \label{lem:rank-f-monomial}
        f = \sum_m g_m(\xbar) m,
    \end{equation}
    where the sum ranges over all multilinear monomials $m$ in the $\ybar$-variables and $g_m(\xbar)$ is some multilinear polynomial in the $\xbar$-variables.

    \begin{claim}
        For any monomial $m$ that is set-multilinear on some $w|_T$, where $T \subseteq N_w$, the leading monomial of $g_m(\xbar)$ is less than or equal to
        \[
        m (\sigma(m)|_{A_w^{S}}),
        \]
        where $S$ is the maximal subset of $P_w$ such that $A_w^S \subseteq B_w^T$.
        Moreover, if $m$ is set-multilinear on $w|_{N_w}$, then the leading monomial of $g_m(\xbar)$ equals 
        \[
        m (\sigma(m)|_{A_w^{P_w}}).
        \]
    \end{claim}
    \begin{proof}
        We prove this claim by induction on the size of $T$.\\\\
        \textit{Base case:} If $T = \emptyset$, consider the partial assignment $\tau_1$ that maps all the $\ybar$-variables to $0$.
        We have $\tau_1(f) = g_1(\xbar)$, where $g_1(\xbar)$ is the coefficient of the empty monomial $1$.
        On the other hand, $\tau_1(\ks_{w,p}^{(i)}) = 1$ for all $i$.
        Since
        \[
            f = \frac{1}{\ks_{w,p}} \text{ over Boolean assignments,}
        \]
        we see that $\tau_1(f) = 1/(r-\beta)$ over Boolean assignments.
        As $g_1(\xbar)$ is multilinear, $g_1(\xbar) = 1/(r-\beta)$ as a polynomial identity, so the the leading monomial of $g_1(\xbar)$ is the empty monomial $1$.
        \\\\
        \textit{Inductive step:}
        Suppose that $T$ is non-empty, and let $m$ be a set-multilinear monomial over $w|_T$.
        Consider the partial assignment $\tau_m$ that maps any $\ybar$-variable in $m$ to $1$ and any other $\ybar$-variable to $0$.
        By \Cref{lem:rank-f-monomial}
        \begin{equation}
        \label{eqn_f_inductive_assign}
            \tau_m(f) = \sum_{m'}g_{m'}(\xbar),
        \end{equation}
        where $m'$ ranges over all submonomials of $m$.
        On the other hand,
        \[
        \tau_m(\ks_{w,p}) = \sum_{j \in [r]} \prod_{i \in P_w^{(j)}} \tau_m(\ks_{w,p}^{(i)}) - \beta.
        \]
        For $i \in P_w$, if $A_w^{(i)} \not \subseteq B_w^T$, then $\tau_m(\ks_{w,p}^{(i)}) = 1$; however, if $A_w^{(i)} \subseteq B_w^T$, then $\tau_m(\ks_{w,p}^{(i)}) = 1- x_{\sigma_i}^{(i)}$, where $\sigma_i$ is the binary string indexed by $A_w^{(i)}$ that agrees with $\sigma(m)$ on $A_w^{(i)}$.
        Therefore
        \[
        \tau_m(f)  \left(\sum_{j \in [r]} \prod_{i \in P_w^{(j)} } (1- x_{\sigma_i}^{(i)}) - \beta \right)= 1  \text{ over Boolean assignments,}
        \]
        where the product ranges over $i \in P_w^{(j)}$ such that $A_w^{(i)} \subseteq B_w^T$.
        From \Cref{lem:deg-bound-special}, it follows that the leading monomial of $\tau_m(f)$ is the product of all the $x_{\sigma_i}^{(i)}$ appearing above, and thus the leading monomial is
        \begin{equation}
        \label{eqn_lead_monomial}
            m (\sigma(m)|_{A_w^{S}}),
        \end{equation}
        where $S$ is the maximal subset of $P_w$ such that $A_w^S \subseteq B_w^T$ (see \Cref{fig:leading-monomial}).
        If the leading monomial of $g_m(\xbar)$ were greater than $ m (\sigma(m)|_{A_w^{S}})$, then it must be cancelled by some monomial of $g_{m'}(\xbar)$ in \Cref{eqn_f_inductive_assign} for some proper submonomial of $m$; however, by the inductive hypothesis, for all such proper submonomials $m'$, the leading monomial of $g_{m'}(\xbar)$ is less than or equal to $m (\sigma(m')|_{A_w^{S}})$.
        Therefore, the leading monomial of $g_m(\xbar)$ must be less than or equal to \Cref{eqn_lead_monomial}, concluding the induction.\\\\
        It remains to show that the leading monomial of $g_m(\xbar)$ equals  $m (\sigma(m)|_{A_w^{P_w}})$ whenever $m$ is set-multilinear on $w|_{N_w}$.
        Let $m'$ be a proper submonomial of $m$ that is set-multilinear over $w|_{T}$ for some $T \subsetneq N_w$.
        As $w$ is a balanced word, there is some $i \in P_w$ such that $A_w^{(i)} \not\subseteq B_w^T$, and thus the leading monomial of $g_{m'}(\xbar)$ is strictly smaller than $m (\sigma(m)|_{A_w^{P_w}})$.
        From \Cref{eqn_lead_monomial}, it follows that the leading monomial of $g_m(\xbar)$ must equal $m (\sigma(m)|_{A_w^{P_w}})$.
    \end{proof}
    For each monomial $m_P$ that is set-multilinear over $w|_{P_w}$, there exists a monomial $m_N$, set-multilinear over $w|_{N_w}$, such that the leading monomial of $g_{m_N}(\xbar)$ is exactly $m_P$.
    Consequently, the $(m_P,m_N)$ entry of $M_w(f)$ is non-zero in $\FF$, while for every monomial $m_P' \neq m_P$, also set-multilinear over $w|_{P_w}$ and satisfying $m_P \leq m_P'$, the $(m_P',m_N)$ entry is zero.
    For $M_w(f)$, it follows that the dimension of the column space equals the number of rows, so $M_w(f)$ has full rank.
\end{proof}

\begin{corollary}
\label{cor:relrank-lower-bound}
    Let $\mathbb{F}$ be a field with characteristic $p$, and let $w \in \mathbb{Z}^d$ be a balanced word with $|w_i| \leq b$ for all $i \in [d]$.
    If $f$ is the multilinear polynomial such that
    \[
    f = \frac{1}{\ks_{w,p}} \text{ over Boolean assignments,}
    \]
    then $\relrk(f) \geq 2^{-b/2}$.
\end{corollary}
\begin{proof}
    Recall that, by the construction of $\ks_{w,p}$, we assume $|w_{N_w}| \geq |w_{P_w}|$.
    Since $w$ is balanced and satisfies $|w_i| \leq b$ for all $i \in [d]$, it follows that $|w_{P_w}| - |w_{N_w}| \geq -b$.
    By \Cref{lem:rank-lower-bound}, $M_w(f)$ has rank $|M^P_w|$.
    Therefore
    \[
        \relrk_w(f) = \sqrt{\frac{|M^P_w|}{|M^N_w|}} = \sqrt{2^{|w_{P_w}| - |w_{N_w}|}} \geq 2^{-b/2}.
    \]
\end{proof}

\begin{figure}
    \centering
    \includegraphics[scale=0.7]{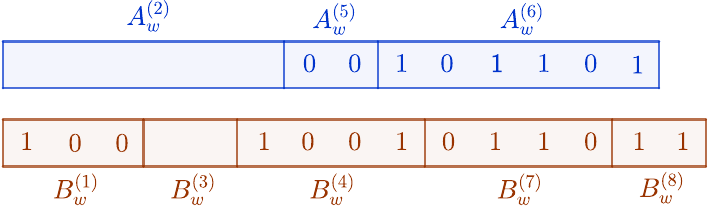}
    \caption{\small From \cite{GHT22}. In this example, $T = \{1,4,7,8\} \subseteq N_w$ and  $m = y_{100}^{(1)} \cdot y_{1001}^{(4)} \cdot y_{0110}^{(7)} \cdot y_{11}^{(8)}$ is a set-multilinear monomial over $w \vert_T$. Like \cite{GHT22}, since $S = \{5,6\}$ is the maximal subset of $P_w$ with $A_w^S \subseteq B_w^T$, we have that the leading monomial of $g_m(\vx)$  is less than or equal to $x_{00}^{(5)} \cdot x_{101101}^{(6)}$. However, in contrast to \cite{GHT22}, in our polynomial $\ks_{w,p}$, the partial assignment setting the $\vy$-variables in $m$ to $1$ and the remaining $\vy$-variables to $0$ results in the polynomial $(1-x_{00}^{(5)}) + (1-x_{101101}^{(6)}) - \beta$.}
    \label{fig:leading-monomial}
\end{figure}

\subsection{IPS Lower Bound}
We now state and prove our lower bound for constant-depth $\IPS$ over finite fields.
We begin by recalling notation from \cite{BDS24_journal}.
Let $F(n)$ denote the $n$-th Fibonacci number, defined by $F(0)=1, F(1)=2$ and $F(i)=F(i-1)+F(i-2)$ for $i \geq 2$; let $G(i) = F(i) - 1$ for all $i$.
Fix the product-depth $\Delta \leq \nicefrac{\log \log \log n}{4}$, and let $d = \lfloor \nicefrac{\log n}{4}  \rfloor$ and $\lambda = \lfloor d^{1/G(\Delta)} \rfloor$.
\begin{theorem}[\cite{GHT22} over Finite Fields]
\label{thm:ips-modp-main}
    Let $p \geq 5$ be a prime, and let $\FF$ be a field of characteristic $p$.
    Let $n, \Delta \in \NN_+$ with $\Delta \leq \nicefrac{\log \log \log n}{4} $.
    Then any product-depth at most $\Delta$ multilinear $\lbIPS$ refutation over $\FF$ of  $\ks_{w,p}$ has size  at least
    \[
    n^{\Omega(\lambda/\Delta)}.
    \]
\end{theorem}
The proof of \Cref{thm:ips-modp-main} relies on the following result:
\begin{theorem}
\label{thm:knapsack-modp-lower-bound}
    Let $p \geq 5$ be a prime, and let $\FF$ be a field of characteristic $p$.
    Let $\Delta$ be as above.
    If $f$ is the multilinear polynomial that equals
    \[
    \frac{1}{\ks_{w,p}} \text{ over Boolean assignments,}
    \]
    then any circuit of product-depth at most $\Delta$ computing $f$ has size at least
    \[
    n^{\Omega(\lambda/\Delta)}.
    \]
\end{theorem}
\begin{proof}[Proof of \Cref{thm:ips-modp-main} from \Cref{thm:knapsack-modp-lower-bound}]
Let $C(\vx, \vy, \vz)$ be a multilinear $\lbIPS$ refutation of $\ks_{w,p}$\footnote{$\ks_{w,p}$ involves both $\vx$- and $\vy$-variables. As this distinction will not play a role in the proof, we treat all variables in $\ks_{w,p}$ as $\vx$-variables. We therefore use the standard notation $C(\vx, \vy, \vz)$ for an $\IPS$ refutation, where $\vx$ are variables of the axioms, and $\vy$, $\vz$ serve as placeholder variables for the axioms.}.
As there is only one non-Boolean axiom, $C$ has a single $\ybar$-variable, which we denote by $y$.
Since $\widehat C(\vx, y, \vnz)$ is linear in the $y$-variable and satisfies $\widehat C(\vx, 0, \vnz) = 0$, it follows that
\[
\widehat C(\vx, y, \vnz) = g(\vx) \cdot y
\]
for some polynomial $g(\vx) \in \FF[\vx]$.
This polynomial $g(\vx)$ is computed by the circuit $C(\vx,1,\vnz)$, so the minimal product-depth-$\Delta$ circuit size of $g(\vx)$ lower bounds that of $C(\vx, \vy, \vz)$.
Therefore, it suffices to lower bound the size of product-depth at most $\Delta$ circuits computing $g(\vx)$.

We have
        $$
        \widehat C(\vx,y,\vz)= \widehat  C(\vx,y,\vnz) +\sum_{i} h_i\cdot z_i
        $$
        for some polynomials $h_i$ in $\vx,y,\vz$,
        hence $\widehat  C(\vx,y,\vz) =  g(\vx)\cdot y + \sum_{i} h_i\cdot z_i$. 
        Since
        \[
        \widehat C\left(\vx,\ks_{w,p},\vx^2-\vx\right) = 1, 
        \]
        we see that
        \[
        g(\vx) \cdot \left(\ks_{w,p}\right) + \sum_{i} (h_i \cdot (x_i^2-x_i)) = 1.
        \]
        Therefore, over Boolean assignments, $g(\vx)\cdot\ks_{w,p} \equiv 1$.
        The result now follows from \Cref{thm:knapsack-modp-lower-bound}.
\end{proof}
 
\begin{lemma}
\label{lem:set-multilinear-circuit-lb}
Let $p \geq 5$ be a prime, and let $\FF$ be a field of characteristic $p$.
Let $\Delta$, $d$ and $\lambda$ be as above.
There exist $\alpha \in \mathbb{Q}$ with $\nicefrac{1}{2} \leq \alpha < 1$, and $k \in \NN_+$ with $k \in [\nicefrac{\lfloor \log n \rfloor}{2}, \lfloor \log n \rfloor]$ and $\alpha k \in \ZZ$, such that if $w \in \ZZ^d$ is a balanced word over the alphabet $\{\alpha k, -k\}$, and $f$ is the multilinear polynomial which equals $\nicefrac{1}{\ks_{w,p}}$ over Boolean assignments, then any set-multilinear circuit of product-depth $\Delta$ computing the set-multilinear projection $\Pi_w(f)$ has size at least
    \[
    s \geq 2^{\frac{k(\lambda/256 - 1)}{2 \Delta}}.
    \]
\end{lemma}
\begin{proof}
Let $C$ be a set-multilinear circuit of size $s$ and product-depth $\Delta$ computing $\Pi_w(f)$.
By unwinding $C$ into a formula, we obtain a set-multilinear formula $F$ of size $s^{2 \Delta}$ and product-depth $\Delta$ that also computes $\Pi_w(f)$.
We now make use of the following claim from \cite{BDS24_journal}:
\begin{claim}[\cite{LST21},\cite{BDS24_journal} Lemma 4.3]
\label{clm_rel_rank_bound}
    Let $\delta \leq \Delta$ be an integer.
    There exist $\alpha \in \mathbb{Q}$ with $\nicefrac{1}{2} \leq \alpha < 1$, and $k \in \NN_+$ with $k \in [\nicefrac{\lfloor \log n \rfloor}{2}, \lfloor \log n \rfloor]$ and $\alpha k \in \ZZ$, such that if $w \in \ZZ^d$ is a word over the alphabet $\{\alpha k, -k\}$, and $F$ is a set-multilinear formula of product-depth $\delta$, degree at least $\lambda^{G(\delta)}/8$ and size at most $s$, then 
    \[
    \relrk_w(F) \leq s 2^{-k \lambda/256}.
    \]
\end{claim}
As $w$ is balanced, by \Cref{lem:rank-lower-bound}, $M_w(f)$ has full rank and $\deg F \geq d \geq \lambda^{G(\delta)}/8$.
Thus, applying \Cref{cor:relrank-lower-bound} and \Cref{clm_rel_rank_bound}, we obtain 
\[
2^{-k} \leq \relrk_w(\Pi_w(f)) \leq s^{2 \Delta} 2^{-k \lambda/256}.
\]
We therefore see that
\[
    s^{2 \Delta} \geq 2^{k(\lambda/256 - 1)},
\]
from which the claim of the lemma follows.
\end{proof}

\begin{lemma}[\protect{\cite[Corollary~27]{forbes24}}]
    \label{lem:forbes-set-multilinearization}
    Let $\mathbb{F}$ be any field, and let the variables $\xbar$ be partitioned into $\xbar = \xbar_1 \sqcup \cdots \sqcup \xbar_d$.
    Suppose $f \in \mathbb{F}[\xbar]$ can be computed by a size $s$, product-depth $\Delta$ algebraic circuit.
    Then the set-multilinear projection $\Pi_{\textnormal{sml}}(f) \in \mathbb{F}[\xbar]$ can be computed by a size $\poly(s, \Theta (\frac{d}{\ln d} )^d)$, product-depth $2\Delta$ set-multilinear circuit.    
\end{lemma}

\begin{proof}[Proof of \Cref{thm:ips-modp-main}]
    Let $C$ be a circuit of size $s \geq n$ and product-depth at most $\Delta$ computing $f$.
    Let $d = \lfloor \nicefrac{\log n}{4}  \rfloor$ and $\lambda = \lfloor d^{1/G(\Delta)} \rfloor$ be as defined above, and let $\nicefrac{1}{2} \leq \alpha < 1$ and $k \in [\nicefrac{\lfloor \log n \rfloor}{2}, \lfloor \log n \rfloor]$ be as constructed in \Cref{lem:set-multilinear-circuit-lb}.
    Construct, by induction, a balanced word $w \in \ZZ^d$ over the alphabet $\{\alpha k, -k\}$.
    
    By \Cref{lem:forbes-set-multilinearization}, there exists a set-multilinear circuit $C'$ of size $\poly(s, \Theta (\frac{d}{\ln d} )^d)$ and product-depth $2\Delta$ computing the set-multilinear projection $\Pi_w(f)$ of $f$.

    Moreover, by \Cref{lem:set-multilinear-circuit-lb}, any set-multilinear circuit of product-depth $2\Delta$ computing $\Pi_w(f)$ must have size at least
    \[
        2^{\frac{k(\lambda/256 - 1)}{2 \Delta}} \geq n^{\frac{\lambda/256 - 1}{8 \Delta}},
    \]
    where the inequality follows from the lower bound on $k$.
    Combining the two bounds above, we obtain
    \[
    \poly(s, \Theta (\frac{d}{\ln d} )^d) \geq n^{\frac{\lambda/256 - 1}{8 \Delta}},
    \]
    and therefore,
    \[
    d ^{O(d)} \poly(s) \geq n^{\Omega(\lambda/\Delta)}.
    \]
    Since $\Delta \leq \nicefrac{\log \log \log n}{4} \leq \nicefrac{\log \log d}{2}$, it follows that $\lambda \geq (\log d)^2$.
    Hence,
    \[
    n^{\Omega(\lambda/\Delta)} \geq d^{\omega(d)},
    \]
    from which the claim of the theorem follows.
\end{proof}

\textbf{Comment}
(constructing a scattered partition):
Our instance $\ks_{w,p}$ and rank lower bound require a scattered partition $P_w = P_w^{(1)} \sqcup \cdots\sqcup P_w^{(r)}$ of the positive indices having fewer than $p$ parts.
Let $\xi := \Delta_G(N_w)$ be the negative overlap of the overlap graph.
We construct a scattered partition with $r = \xi$ as follows.
For each positive index $i \in P_w$, define $\pi'(i) := |\{i' \in P_w \mid i' \leq i\}|$, and let $\pi(i)$ be the least residue of $\pi'(i) \pmod{\xi}$; that is, $\pi(i) \in [\xi]$ with $\pi(i) \equiv \pi'(i) \pmod{\xi}$.
The map $\pi$ partitions $P_w$ into $\xi$ parts such that in each part, the neighbourhoods of the vertices are pairwise disjoint, hence $\pi$ induces a scattered partition (see \Cref{fig:scattered-partition}).

Because our $\IPS$ lower bound assumes $p \geq 5$, it suffices to construct a scattered partition with $r < 5$.
Since the word $w \in \ZZ^d$ is over the alphabet $\{\alpha k, -k\}$ with  $\nicefrac{1}{2} \leq \alpha < 1$, it follows that $ \Delta_G(N_w) \leq 3$.
Therefore, $\pi$ yields a scattered partition with $r \leq 3$.

\begin{figure}
    \centering
    \includegraphics[scale=0.7]{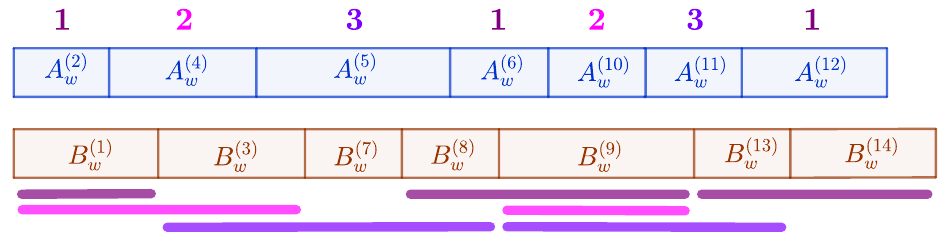}
    \caption{\small Illustration of the scattered partition induced by $\pi:P_w \to [\Delta_G(N_w)]$ with $\Delta_G(N_w) = 3$. The values of $\pi$ appear above the positive boxes, while the neighbourhoods of the positive indices are shown below the negative boxes. Since this is a scattered partition, vertices in the same part have pairwise disjoint neighbourhoods.}
    \label{fig:scattered-partition}
\end{figure}

\section{Upper Bounds for Constant-depth Multilinear IPS}
\label{sec:upper_bound}
\subsection{Elementary Symmetric Sums}

\begin{proposition}
Over any field $\mathbb{F}$, for $|\xbar| = n \geq l \geq d \geq 0$,
    \[ e_l (\xbar) \cdot e_d (\xbar) = \sum_{i=k}^{d} \binom{l+d - i}{l} \binom{l}{i} e_{l+d - i} (\xbar) \pmod {\xbar^2 - \xbar},\]
where $k \geq 0$ is the smallest integer such that $l+d - k \leq n$. 
\end{proposition}
\begin{proof}
    Since $\ml(e_l (\xbar) \cdot e_d (\xbar))$ is symmetric in $\xbar$, we have
    \[
        e_l (\xbar) \cdot e_d (\xbar) = \sum_{i=k}^{d} \gamma_i \cdot e_{l+d - i} (\xbar) \pmod {\xbar^2 - \xbar},
    \]
    for $\gamma_i \in \mathbb{F}$.
    Let $S \subseteq [n]$ with $|S| = l+d - i$.
    The coefficient of the monomial $x^S$ in $\ml(e_l (\xbar) \cdot e_d (\xbar))$ is $\binom{l+d - i}{l} \binom{l}{i}$.
    This is because each of the $\binom{l+d - i}{l}$ many sub-monomials $x^A$ of $x^S$ in $e_l (\xbar)$ combines with $\binom{l}{i}$ many monomials $x^B$ in $e_d (\xbar)$, where $|A \cap B| = i$, to produce $x^S$ in $\ml(e_l (\xbar) \cdot e_d (\xbar))$.
    Therefore, $\gamma_i = \binom{l+d - i}{l} \binom{l}{i}$.
\end{proof}

We observe that elementary symmetric polynomials are unsatisfiable even over constant characteristic fields, when their degree meets a simple condition on their $p$-base expansion.

\begin{lemma}[\cite{Lucas_Theorem_paper} Lucas's Theorem]
\label{thm:lucas}
    Let $p$ be a prime and $m,n \in \NN_+$.
    If $m = m_k p^k + \cdots + m_1p + m_0$ and $n = n_k p^k + \cdots + n_1p + n_0$ are the base $p$ expansions of $m$ and $n$ respectively (where $0 \leq m_i,n_i \leq p-1$ for $0 \leq i \leq k$), then 
    \[\binom{m}{n} \equiv \prod_{i=0}^{k} \binom{m_i}{n_i} \pmod p.\]
\end{lemma}

\begin{lemma}
\label{lem:symmetric-sums-boolean}
    Let $\FF$ be a field with characteristic $p$. 
    If $d = d_k p^k + \cdots + d_1p + d_0$ is the base $p$ expansion of $d$, then 
    \[
        | \{e_d(\xbar) \mid x \in \{0,1\}^n \} | \leq \prod_{\substack{i \in \{0, \ldots, k\} \\ d_i \neq 0}} (p-d_i) + 1.
    \]
    In particular, if $d$ has only one non-zero digit $d_i$ in its base $p$ expansion and $d_i \geq 2$, then there exists $\beta \in \mathbb{F}$ such that $e_d(\xbar) - \beta = 0$ is unsatisfiable over Boolean assignments.
\end{lemma}
\begin{proof}
    If $m$ is the Hamming weight of a Boolean assignment $\alpha \in \{0,1\}^n$,
    then $e_d(\alpha) = \binom{m}{d}$.
    By \Cref{thm:lucas}, we have 
    \[
        \binom{m}{d} \equiv \prod_{i \in \{0, \ldots,k\}}\binom{m_i}{d_i} \equiv \prod_{\substack{i \in \{0, \ldots,k\} \\ d_i \neq 0} }\binom{m_i}{d_i} \pmod p.
    \]
    We see that $e_d(\alpha) = 0$ in $\FF$ if and only if $m_i < d_i$ for some $i$ with $d_i \neq 0$.
    Conversely, $e_d(\alpha)$ is non-zero in $\FF$ if and only if $d_i \leq m_i \leq p-1$ for all $i$ with $d_i \neq 0$.
    Therefore, $e_d(\alpha)$ can attain at most
    \[
    \prod_{\substack{i \in \{0, \ldots,k\} \\ d_i \neq 0} }((p-1)-d_i+1)
    \]
    distinct non-zero values in $\FF$.
    This completes the proof of the main claim of the lemma.

    Now if $d_i \geq 2$ is the only non-zero digit in the base $p$ expansion of $d$, then over Boolean assignments, $e_d(\xbar)$ can attain at most $p - d_i + 1 \leq p-1$ distinct values in $\FF$.
    The existence of $\beta$ follows. 
\end{proof}

Occasionally, instead of viewing $\xbar$ as Boolean variables, we consider them more generally as Boolean functions.
By a similar argument, it is straightforward to verify that \Cref{lem:symmetric-sums-boolean} continues to hold in this more general setting.

\subsection{Separation}

Here, we separate the constant-depth $\IPS$ subsystem over finite fields, as considered in this work, from the constant-depth $\IPS$ subsystem over large fields studied in \cite{GHT22}.

Let $p \geq 3$ be a prime and let $\F$ be a field of characteristic $p$.
We construct the \textit{symmetric knapsack} of degree $2$, denoted as $\ks_{w, e_2}$.
Over Boolean assignments, $\ks_{w, e_2}$ is unsatisfiable in $\F$ and in every field of characteristic $0$.
Moreover, constant-depth multilinear $\lbIPS$ over $\F$ admits a polynomial-size refutation of $\ks_{w, e_2}$, whereas constant-depth multilinear $\lbIPS$ over any characteristic $0$ field does not.

Using the notation from \Cref{sec:cd-IPS}, let $w \in \ZZ^d$ be a word with $|w_i| \leq b$ for every $i$.  
Our separating instance is defined as:
\[
    \ks_{w, e_2} := \ml \left(e_2 \left(\{x_\sigma^{(i)}f_\sigma^{(i)}\}_{i\in P_w, \sigma\in\{0,1\}^{A_w^{(i)}}} \right) \right) - \beta.
\]
where $\beta \in \FF$ is chosen such that $\ks_{w, e_2} = 0$ admits no satisfying Boolean assignment in $\FF$.

\textbf{Comment} (the existence of $\beta$):
The existence of $\beta$ follows from \Cref{lem:symmetric-sums-boolean}, specifically from the remark following the lemma concerning its application to Boolean functions rather than Boolean variables.

\textbf{Comment} (computing $\ks_{w,e_2}$ by a $\poly(d,2^b)$-size, product-depth $2$, multilinear formula):
Since
\[
    \ks_{w,e_2} = \sum_{\substack{(i_1,\sigma_1), (i_2, \sigma_2) \in S \\ (i_1,\sigma_1)\neq (i_2, \sigma_2)}} x_{\sigma_1}^{(i_1)}x_{\sigma_2}^{(i_2)} \ml(f_{\sigma_1}^{(i_1)}f_{\sigma_2}^{(i_2)}) - \beta
\]
where $S = \{(i, \sigma) \mid i\in P_w, \sigma\in\{0,1\}^{A_w^{(i)}}\}$, it suffices to verify that $\ml(f_{\sigma_1}^{(i_1)}f_{\sigma_2}^{(i_2)})$ can be computed by a suitable polynomial-size constant-depth multilinear formula.
Since the positive overlap of $w$ satisfies $\Delta_G(P_w) \leq b$, we see that each $f_{\sigma}^{(i)}$ can be written as $\sum \prod (1-y)$ where the fan-in of the sum gate is $O(2^b)$ and the product ranges over distinct $\ybar$-variables.
Moreover, because $\ml((1-y)^2) = 1-y$, each $\ml(f_{\sigma_1}^{(i_1)}f_{\sigma_2}^{(i_2)})$ can likewise be written in this form.
Altogether, $\ks_{w,e_2}$ can thus be written as a product-depth $2$, multilinear formula of $\poly(d,2^b)$-size.
Moreover, each $f_\sigma^{(i)}$ has degree at most $O(b2^b)$, so the overall degree of $\ks_{w,e_2}$ is therefore $O(b2^b)$.

\begin{lemma}
     Over $\FF$, there exists a product-depth $3$ multilinear $\lbIPS$ refutation of $\ks_{w,e_2}$ of size  $\poly(d^p,2^{bp})$.
\end{lemma}
\begin{proof}
    By Fermat's little theorem, we see that
    \begin{equation}
    \label{eq:sym-ks-deg-2-refutation}
        \ml((\ks_{w,e_2})^{p-2}) \cdot \ks_{w,e_2} + \sum_{\psi \in \xbar \cup \ybar} h_\psi (\psi^2 - \psi) = 1
    \end{equation}
    for some polynomials $h_\psi \in \FF[\xbar,\ybar]$.
    From computing $\ks_{w,e_2}$ by a $\poly(d,2^b)$-size product-depth $2$, multilinear formula, we see that $\ml((\ks_{w,e_2})^{p-2})$ can be computed by a product-depth $2$, multilinear formula of size $\poly(d^p,2^{bp})$.
    Moreover, we see that each $h_\psi$ can be computed by a product-depth $2$ formula of size $\poly(d^p,2^{bp})$.
    Therefore, over $\FF$, \Cref{eq:sym-ks-deg-2-refutation} is a product-depth $3$ multilinear $\lbIPS$ refutation of $\ks_{w,e_2}$. 
\end{proof}

Let $E$ be a field of characteristic $0$.
Since $\ks_{w, e_2}$ admits no satisfying Boolean assignment in $\F$, it likewise admits none in $E$.
Over $E$, we will prove a lower bound against constant-depth multilinear $\lbIPS$ for $\ks_{w, e_2}$.
We first prove a degree lower bound.
\begin{lemma}
\label{lem:deg-lower-bound-e2-large-field}
    Let $\textnormal{char}(\mathbb{F}) = 0$, $n > 1$ and $\beta \in \ZZ^+$ such that $e_2 (x_1, \ldots, x_n) - \beta = 0$ is unsatisfiable over $\mathbb{F}$ for $\xbar \in \{0,1\}^n$.
    If $f \in \mathbb{F}[x_1, \ldots, x_n]$ is the multilinear polynomial such that
    \[
    f(\xbar) \left(e_2 (\xbar) - \beta \right) = 1 \pmod {\xbar^2 - \xbar},
    \]
    then $\deg f = n$.
\end{lemma}
We note that a degree lower bound of $ \deg f \geq n-1$ follows from \cite{HLT24} Corollary 1.2; however, we need the tight bound of $\deg f \geq n$.
\begin{proof}[Proof of \Cref{lem:deg-lower-bound-e2-large-field}]
    As $f(\xbar)$ is multilinear, we have
    \[
    f(\xbar) = \sum_{T \subseteq [n]} f(\mathbbm{1}  _T) \prod_{i \in T} x_i \prod_{i \notin T}(1-x_i),
    \]
    where $\mathbbm{1}_T \in \{0,1\}^n$ is the indicator vector of the set $T$.
    Therefore
    \[
    f(\xbar) = \sum_{T \subseteq [n]} \frac{1}{\binom{|T|}{2} - \beta} \prod_{i \in T} x_i \prod_{i \notin T}(1-x_i).
    \]
    The coefficient of $\prod_{i \in [n]} x_i$ in $f(\xbar)$ is thus
    \begin{equation}
    \label{eqn_lead_coef_e2-b}
        \sum_{T \subseteq [n]} \frac{1}{\binom{|T|}{2} - \beta} (-1)^{n - |T|} = \sum_{j = 0}^n \binom{n}{j} \frac{1}{\binom{j}{2} - \beta} (-1)^{n - j}.
    \end{equation}
    We show that \Cref{eqn_lead_coef_e2-b} is nonzero.
    As \Cref{eqn_lead_coef_e2-b} lies in the subfield of $\mathbb{F}$ that is isomorphic to $\mathbb{Q}$, it suffices to show that it is nonzero over $\mathbb{Q}$.
    It therefore suffices to show that \Cref{eqn_lead_coef_e2-b} is nonzero over $\mathbb{R}$.
    We have, over $\mathbb{R}$,
    \[
        \frac{1}{\binom{j}{2} - \beta} = \frac{2}{j^2 - j - 2\beta} = \frac{2}{(j-\gamma_1)(j-\gamma_2)} = 2\sqrt{1+8\beta} \left(\frac{1}{j-\gamma_2} - \frac{1}{j-\gamma_1}\right),
    \]
    where $\gamma_1 = (1 - \sqrt{1+8\beta})/2$ and $\gamma_2 = (1 + \sqrt{1+8\beta})/2$.
    Hence
    \begin{align*}
        \sum_{j = 0}^n \binom{n}{j} \frac{1}{\binom{j}{2} - \beta} (-1)^{n - j} &= 2\sqrt{1+8\beta} \sum_{j = 0}^n \binom{n}{j}\left(\frac{1}{j-\gamma_2} - \frac{1}{j-\gamma_1}\right) (-1)^{n - j} \\
        &= 2\sqrt{1+8\beta} \left(- \frac{n!}{\prod_{j=0}^n (\gamma_2 - j)} + \frac{n!}{\prod_{j=0}^n (\gamma_1 - j)} \right),
    \end{align*}
    where the last equality follows from
    \begin{claim}[\cite{FSTW21} Subclaim B.2]
        \[
        \sum_{j=0}^k \binom{k}{j} \frac{1}{j-\beta} (-1)^{k-j} = - \frac{k!}{\prod_{j=0}^k (\beta - j)}.
        \]
    \end{claim}
Finally, to show that \Cref{eqn_lead_coef_e2-b} is nonzero, it suffices to show that
\begin{equation}
\label{eqn_lead_coef_e2-b_aux_prod}
    \prod_{j=0}^n (\gamma_1 - j) \neq \prod_{j=0}^n (\gamma_2 - j).
\end{equation}
We note that $\gamma_1 (\gamma_1 - 1) = \gamma_2 (\gamma_2 - 1) = 2\beta$; however, for $k > 1$, we have $|\gamma_1 - k| > |\gamma_2 - k|$.
We therefore have
\[
    \left|\prod_{j=0}^n (\gamma_1 - j) \right| > \left|\prod_{j=0}^n (\gamma_2 - j) \right|,
\]
hence \Cref{eqn_lead_coef_e2-b_aux_prod} holds and \Cref{eqn_lead_coef_e2-b} is nonzero.
\end{proof}

\begin{lemma}
    Let $E$ be a field of characteristic $0$, and $n,\Delta \in \NN_+$ with $\Delta \leq \nicefrac{\log \log \log n}{4}$.
    Then any product-depth at most $\Delta$ multilinear $\lbIPS$ refutation of $\ks_{w,e_2}$ is of size at least
    \[
        n^{\Omega(\lambda/\Delta)},
    \]
    where $d = \lfloor \nicefrac{\log n}{4} \rfloor$ and $\lambda = \lfloor d^{1/G(\Delta)} \rfloor$.
\end{lemma}
\begin{proof}
    The proof of this lemma is essentially the same as the proof of \Cref{thm:ips-modp-main} (and \cite{GHT22}) and is omitted here.
    We recall the general strategy of reducing an $\IPS$ lower bound to a rank lower bound, to a degree lower bound, which for this instance is \Cref{lem:deg-lower-bound-e2-large-field}).
\end{proof}

\begin{theorem}[Separation: \cite{GHT22} over Finite Fields vs.\ \cite{GHT22}]
\label{thm:separation}
    Let $p \geq 3$ be a prime, and let $\F$ be a field of characteristic $p$.
    Let $n, \Delta \in \NN_+$ with $\Delta \leq \nicefrac{\log \log \log n}{4} $.
    Then there exists a product-depth $2$, multilinear formula $\ks_{w,e_2}$ of size $\poly(n)$ such that:
    \begin{itemize}
        \item $\ks_{w,e_2}$ has no satisfying Boolean assignment over $\F$, and over any field of characteristic $0$;
        \item there is a $\poly(n)$-size, product-depth-$3$ multilinear $\lbIPS$ refutation of $\ks_{w,e_2}$ over $\F$;
        \item for every field of characteristic $0$, any product-depth at most $\Delta$ multilinear $\lbIPS$ refutation of $\ks_{w,e_2}$ requires size at least
    \[
        n^{\Omega(\lambda/\Delta)},
    \]
    where $d = \lfloor \nicefrac{\log n}{4} \rfloor$ and $\lambda = \lfloor d^{1/G(\Delta)} \rfloor$.
    \end{itemize}
\end{theorem}

\section{Lower Bounds for roABP-IPS}
\label{sec:roABP}
In this section, we work over field $\F_q$, where $q$ is a constant greater than 2. In addition to proving a lower bound over finite fields, this work significantly simplifies \cite{HLT24}, though we note that we are still working in the placeholder model. Consider the following hard instance
\begin{equation}
    \label{equa: hard instance for roABP}
    f(\xbar) \coloneqq \prod_{i=1}^n (1 - x_i) - 2.
\end{equation}
Clearly this function never evaluates to 0 over boolean assignments in $\F_q$ (when $q>2$). In contrast, the hard instance from \cite{HLT24} is a subset-sum instance, therefore requiring large characteristic to be defined. 

\subsection{roABP-IPS Lower Bounds in Fixed Order}
We begin by proving a lower bound where the roABPs are given a fixed order of the variables. By \Cref{lem:deg-bound-special}, we have the following lemma.

\begin{lemma}
    \label{lemma: degree lower bounds for roABP}
    Let $\F_q$ be a finite field with $q >2$, and let $f(\xbar)$ be as in \Cref{equa: hard instance for roABP}. If $g(\xbar)$ is the multilinear polynomial such that
    \[
        g(\xbar) \cdot f(\xbar) = 1 \mod{\xbar^2- \xbar},
    \]
    then $\deg(g) = n$.
\end{lemma}

For any $\xbar, \ybar$ variables, with $|\xbar| = |\ybar|= n$, we use $\xbar \circ \ybar$ to denote the entry-wise product $(x_1 y_1, \dots, x_n y_n)$. In other words, the \emph{gadget} we use is the mapping
\[
    x_i \mapsto x_i y_i,
\]
which substitutes the variables $x_i$ by $x_i y_i$, for every $i$. We use $\mathbbm{1}_S \in \{0,1\}^n$ to denote the indicator vector for a set $S$.

\begin{theorem}
    \label{theorem: leading monomials lower bounds}
    Let $f(\xbar)$ be as in \Cref{equa: hard instance for roABP}. Let $g(\xbar,\ybar) \cdot f(\xbar \circ \ybar) =1 \mod{\xbar^2- \xbar}$. Then,
    \begin{equation}
        \label{equa: leading monomial lower bounds}
        \Big | \LM \big ( \{ \ml(g(\xbar,\mathbbm{1}_S)) : S\subseteq [n]\}\big)\Big| = 2^{n}.
    \end{equation}
    
\end{theorem}

\begin{proof}
We first need the claim below.
    \begin{claim} Each $S \subseteq [n]$ induces a distinct leading monomial in $\ml(g(\xbar, \mathbbm{1}_S))$.
    \end{claim}

    \begin{proof}[Proof of claim:]
        Let $S \subseteq [n]$. By the assumption $g(\xbar, \ybar) \cdot f(\xbar \circ \ybar) =1 \mod{\xbar^2 - \xbar}$, we also have
        \begin{equation}
            \label{equa: condition for degree lower bounds}
            \ml(g(\xbar, \mathbbm{1}_S)) \cdot f(\xbar \circ \mathbbm{1}_S) =1 \mod{\xbar^2-\xbar},
        \end{equation}
         since multilinearizing $g(\xbar, \mathbbm{1}_S)$ does not affect the equality (as we work modulo $\xbar^2 -\xbar$). By the lifting defined above, $\ml(g(\xbar, \mathbbm{1}_S))$ is a (multilinear symmetric) polynomial that \emph{depends on the variables $x_i$, for $i \in S$}. Similarly, $f(\xbar \circ \mathbbm{1}_S)$ is a polynomial of the same form as \Cref{equa: hard instance for roABP} that \emph{depends on the variables $x_i$, for $i \in S$}. In addition, $f(\xbar)$ has no Boolean roots, so neither does $f(\xbar \circ \ybar)$. This together with \Cref{equa: condition for degree lower bounds} means the conditions of \Cref{lemma: degree lower bounds for roABP} are met, so we have
        \[
            \deg(\ml(g(\xbar, \mathbbm{1}_S))) = |S|.
        \]
        Since we assumed that our monomial ordering respects degree,
        \begin{equation}
            \label{equa: leading monomial degree lower bounds}
            \deg(\LM(\ml(g(\xbar, \mathbbm{1}_S)))) = |S|.
        \end{equation}
        There is only one possible multilinear monomial of degree $|S|$ on $|S|$ variables; it follows that every $S$ induces a unique leading monomial (consisting exactly of all variables in $S$).
    \end{proof}
    This concludes the proof of \Cref{theorem: leading monomials lower bounds}.
\end{proof}

\begin{theorem}
    \label{theorem: roABP IPS size lower bounds for fixed order}
    Let $f(\xbar)$ be as in \Cref{equa: hard instance for roABP}. Then, any $\roAlbIPS$refutation of $f(\xbar \circ \ybar) =0$ is of size $2^{\Omega(n)}$, when the variables are ordered such that $\xbar < \ybar$ (i.e., $\xbar$-variables come before $\ybar$-variables).
\end{theorem}

\begin{proof}
    Let $g(\xbar,\ybar)$ be a polynomial such that $g(\xbar, \ybar) \cdot f(\xbar \circ \ybar) =1$ over $\xbar, \ybar \in \{0,1\}^n$. Hence,
    \[
        g(\xbar, \ybar) = \frac{1}{f(\xbar \circ \ybar)} \ \text{ over } \xbar,\ybar \in \{0,1\}^n. 
    \]
    We show that $\dim\coeffs{\xbar| \ybar}g \geq 2^{\Omega(n)}$. This will conclude the proof by Lemma \ref{lem:roABP_width_equals_coeff_dim} which will give the roABP size (width) lower bound and by the functional lower bound in \Cref{thm:func_lb_method}.

    By lower bounding coefficient dimension by the evaluation dimension over the Boolean cube (\Cref{lem:eval_dim_lb_coeff_dim}),
    \begin{align*}
        \dim\coeffs{\xbar|\ybar}g &\geq \dim\evals{\xbar|\ybar,\{0,1\}}g \\
                                                &=\dim\{g(\xbar, \mathbbm{1}_S): S \subseteq [n]\} \\
                                                &\geq \dim\{\ml(g(\xbar,\mathbbm{1}_S)): S\subseteq [n]\}.
    \end{align*}
    Here we used that dimension is non-increasing under linear maps. For $S \subseteq [n]$, denoted by $\xbar_S \coloneqq \{x_i : i \in S\}$ and note that for $\xbar \in \{0,1\}^n$,
    \[  
        g(\xbar, \mathbbm{1}_S) = \frac{1}{f(\xbar_S)},
    \]
    and that $\ml(g(\xbar,\mathbbm{1}_S))$ is a multilinear polynomial only depending on $\xbar_S$.
    By \Cref{theorem: leading monomials lower bounds}, we can lower bound the number of distinct leading monomials of $\ml(g(\xbar, \mathbbm{1}_S)$, where $S$ ranges over subsets of $[n]$:
    \[
        \Big | \LM \big ( \{ \ml(g(\xbar,\mathbbm{1}_S)) : S\subseteq [n]\}\big)\Big| =2^{n}.
    \]
    Therefore, we can lower bound the dimension of the above space by the number of leading monomials (\Cref{lem:span}),
    \begin{align*}
        \dim\coeffs{\xbar|\ybar}g &\geq \dim\{\ml(g(\xbar,\mathbbm{1}_S)): S \subseteq [n]\}\\
                                                &\geq \Big | \LM \big ( \{ \ml(g(\xbar,\mathbbm{1}_S)) : S\subseteq [n]\}\big)\Big| \\
                                                &= 2^{n}.
    \end{align*}
\end{proof}

\subsection{roABP-IPS Lower Bounds in Any Order}
We now extend this previous result to roABPs in any variable order. Consider a polynomial $f(\wbar)$ over $m$ variables, where $m = \binom{2n}{2}$ and $\wbar = \{w_{i,j}\}_{i < j \in [2n]}$. We apply the same gadget from \cite{HLT24}, defined by the mapping
\[
    w_{i,j} \mapsto z_{i,j} x_i x_j,
\]
which substitutes the $m$ variables $w_{i,j}$ by $m+2n$ variables $\{z_{i,j}\}_{i <j \in [2n]}, x_1,\dots,x_{2n}$ such that:
\begin{equation}
    \label{equa: hard instance for roABP for any order}
    f^\star(\zbar,\xbar) \coloneqq f(\wbar)_{w_{i,j} \mapsto z_{i,j} x_i x_j},
\end{equation}
where $f(\wbar)_{w_{i,j} \mapsto z_{i,j} x_i x_j}$ means that we apply the lifting $w_{i,j} \mapsto z_{i,j} x_i x_j$ to the $\wbar$ variables.

Let $f \in \F[\xbar,\ybar,\zbar]$. We denote by $f_{\zbar}$ the polynomial $f$ considered as a polynomial in $\F[\zbar](\xbar,\ybar)$, namely as a polynomial whose indeterminates are $\xbar,\ybar$ and whose scalars are from the ring $\F[\zbar]$. We will consider the dimension of a (coefficient) matrix when the entries are taken from the ring $\F[\zbar]$, and where the dimension is considered over the field of rational functions $\F(\zbar)$. Note that for any $\abar \in \F^{\zbar}$, we have $f_{\abar}(\xbar,\ybar) = f(\xbar, \ybar, \abar) \in \F[\xbar,\ybar]$. We reference the following simple lemma.
\begin{lemma}[\cite{FSTW21}]
    \label{lemma: dimension non-increase under assignment}
    Let $f \in \F[\xbar, \ybar, \zbar]$. Then for any $\abar \in \F^{|\overline{z}|}$
    \[
        \dim_{\F(\zbar)}\coeffs{\xbar|\ybar}f_{\zbar}(\xbar,\ybar) \geq \dim_{\F}\coeffs{\xbar|\ybar}f_{\abar}(\xbar,\ybar).
    \]
\end{lemma}

We now prove the proposition below.
\begin{proposition}
    \label{prop: dimesnion lower bounds for any order}
    Let $n \geq 1$, $m = \binom{2n}{2}$, and $\F_q$ be a finite field of constant characteristic $q$. Let $f \in \F[\wbar]$ be as in \Cref{equa: hard instance for roABP}, and $f^\star(\zbar, \xbar)$ be as in \Cref{equa: hard instance for roABP for any order}. Let $g \in \F[z_1,\dots,z_m, x_1, \dots, x_{2n}]$ be a polynomial such that
    \[
        g(\zbar, \xbar)  = \frac{1}{f^\star(\zbar, \xbar)},
    \]
    for $\zbar \in \{0,1\}^m$ and $\xbar \in \{0,1\}^{2n}$. Let $g_{\zbar}$ denote $g$ as a polynomial in $\F[\zbar][\xbar]$. Then, for any partition $\xbar = (\ubar, \vbar)$ with $|\ubar| = |\vbar| =n$,
    \[
        \dim_{\F(\zbar)}\coeffs{{\ubar|\vbar}}g_{\zbar} \geq 2^{\Omega(n)}.
    \]
\end{proposition}

\begin{proof}
    We embed $\frac{1}{f(\ubar \circ \vbar)}$ in this instance via a restriction of $\zbar$. Define the $\zbar$-evaluation $\abar \in \{0,1\}^{\binom{2n}{2}}$ to restrict $g$ to sum over those $x_ix_j$ in the natural matching between $\ubar$ and $\vbar$, so that
    \begin{equation*}
        \alpha_{ij} = \begin{cases}
            1 \quad x_i=u_k, x_j = v_k,\\
            0 \quad \text{otherwise.}
        \end{cases}
    \end{equation*}
    It follows that $g(\ubar, \vbar, \abar) = \frac{1}{f(\ubar \circ \vbar)}$ for $\ubar, \vbar \in \{0,1\}^n$. Suppose for contradiction that there exists a partition $\xbar = (\ubar, \vbar)$ with $|\ubar| = |\vbar| =n$, such that
    \[
        \dim_{\F(\zbar)}\coeffs{{\ubar|\vbar}}g_{\zbar}(\ubar, \vbar) < 2^{\Omega(n)}.
    \]
    By \Cref{lemma: dimension non-increase under assignment}, we get the relation between the coefficient dimensions of $g_{\zbar}$ and $g_{\abar}$
    \begin{align*}
        \dim_{\F}\coeffs{{\ubar|\vbar}}g_{\abar}(\ubar, \vbar) &\leq \dim_{\F(\zbar)}\coeffs{{\ubar|\vbar}}g_{\zbar}(\ubar, \vbar)\\
            &< 2^{\Omega(n)},
    \end{align*}
    which contradicts our lower bound for a fixed partition (\Cref{theorem: roABP IPS size lower bounds for fixed order}).
\end{proof}

\begin{corollary}
    \label{corollary: roABP lower bounds for any order}
    Let $n \geq 1$, $m = \binom{2n}{2}$, and $\F_q$ be a finite field with constant characteristics $q$. Let $f \in \F[\wbar]$ be as in \Cref{equa: hard instance for roABP}, and $f^\star(\zbar, \xbar)$ be as in \Cref{equa: hard instance for roABP for any order}. Then, any $\roAlbIPS$ refutation (in any variable order) of $f^\star(\zbar, \xbar)$ requires $2^{\Omega(n)}$-size.
\end{corollary}

\begin{proof}
    Consider the polynomial $g \in \F[z_1,\dots,z_m, x_1, \dots, x_{2n}]$ such that
    \[
        g(\zbar, \xbar)  = \frac{1}{f^\star(\zbar, \xbar)}.
    \]
    for $\zbar \in \{0,1\}^m$ and $\xbar \in \{0,1\}^{2n}$. We will show that any roABP computing $g$ requires width $\geq 2^{\Omega(n)}$ in any variable order. The $\roAlbIPS$ lower bound follows immediately from this functional lower bound on $g$ along with the reduction (\Cref{thm:func_lb_method}).

    Suppose that $g(\zbar, \xbar)$ is computable by a width-$r$ roABP in \emph{some} variable order. By pushing the $\zbar$ variables into the fraction field, it follows that $f_{\zbar}$ ($f$ as a polynomial in $\F[\zbar][\xbar]$) is also computable by a width-$r$ roABP over $\F(\zbar)$ in the induced variable order on $\xbar$ (\Cref{fact:roABP}).
    By splitting $\xbar$ in half along its variable order and by the relation between the width of a roABP and its coefficient dimension (\Cref{lem:roABP_width_equals_coeff_dim}), we obtain
    \[
         \dim_{\F(\zbar)}\coeffs{{\ubar|\vbar}}g_{\zbar} < 2^{\Omega(n)},
    \]
     which contradicts the coefficient dimension lower bound of \Cref{prop: dimesnion lower bounds for any order}.
\end{proof}

\subsection{roABP-IPS Lower Bounds by Multiple}
Here we present another roABP-$\IPS$ lower bound over finite fields, but this time using the lower bound for multiples method from \cite{FSTW21}. We introduce the following two lemmas from their paper.

\begin{lemma}[Corollary 6.23 in \cite{FSTW21}]
    \label{lemma: rtABP lower bounds for multiple}
    Let $f \in \F[x_1,\dots, x_n]$ be defined by $f(\xbar) \coloneqq \prod_{i < j} (x_i+x_j +\alpha_{i,j})$ for $\alpha_{i,j} \in \F$. Then for any $0 \neq g \in \F[\xbar]$, $g \cdot f$ requires width-$2^{\Omega(n)}$ as a read-twice oblivious ABP.
\end{lemma}

\begin{lemma}[Lemma 7.1 in \cite{FSTW21}]
    \label{lemma: IPS refutation is a multiple}
    Let $f,\overline{g}, \xbar^2- \xbar \in \F[x_1,\dots, x_n]$ be an unsatisfiable systems of equations, where $\overline{g}, \xbar^2- \xbar$ is satisfiable. Let $C \in \F[\xbar, y, \zbar, \wbar]$ be an $\IPS$ refutation of $f, \overline{g}, \xbar^2-\xbar$. Then
    \[
        1-C(\xbar, 0 , \overline{g}, \xbar^2- \xbar)
    \]
    is a nonzero multiple of $f$.
\end{lemma}

From here, consider the field $\F_q$ for some constant $q$, and let the following two polynomials be our hard system of equations.
\begin{equation}
    f \coloneqq \prod_{i < j} (x_i + x_j +1), \quad
    g \coloneqq \prod_{i=1}^n (1-x_i) -1.
\end{equation}

Note that our $f$ above is the same as in \cite{FSTW21}, but our $g$ differs (they use $\sum_{i=1}^n x_i - n$, which is why they must work in fields of characteristic $>n$). We now state our lower bound.

\begin{theorem}
    \label{theorem: lower bounds by multiple}
    Let $\F_q$ be a finite field of constant characteristic $q$. Let $f,g \in \F[x_1,\dots, x_n]$, where $f \coloneqq \prod_{i < j} (x_i + x_j +1)$ and $g \coloneqq \prod_{i=1}^n (1-x_i) -1$. Then, the system of equations $f, g, \xbar^2- \xbar$ is unsatisfiable, and any $\roAlbIPS$ refutation (in any order of the variables) requires size $\exp\big({\Omega(n)}\big)$.
\end{theorem}

\begin{proof}
    The system $g(\xbar)=0$ and $\xbar^2- \xbar=0$ is satisfiable and has the unique satisfying assignment $\overline{0}$. However, this single assignment does not satisfy $f$ as $f(\overline{0}) = \prod_{i<j} (0+0+1) =1 \neq 0$, so the entire system is unsatisfiable. Thus by \Cref{lemma: IPS refutation is a multiple}, for any $\roAlbIPS$ refutation $C(\xbar,y,z, \wbar)$ of $f, g, \xbar^2- \xbar$, $1- C(\xbar, 0 , g, \xbar^2-\xbar)$ is a nonzero multiple of $f$.
    
    Let $s$ be the size of $C$ as an roABP. We now argue that $1-C(\xbar, 0 ,g, \xbar^2- \xbar)$ has a small read-\emph{twice} oblivious ABP. First, note that we can expand $C(\xbar, 0, z, \wbar)$ into powers of $z$: 
    $$
    C(\xbar, 0 ,z, \wbar)  = C_0(\xbar, \wbar) + C_1(\xbar, \wbar) z.
    $$
    There are only two terms because $C(\xbar, y,z, \wbar)$ is a $\roAlbIPS$ refutation, implying the degree of $z$ in $C(\xbar, y, z, \wbar)$ is at most $1$. Each $C_i(\xbar,\wbar)$ has a $\poly(s)$-size roABP (in the order of the variables of $C$ where $z$ is omitted), as we can compute $C_i$ via interpolation over $z$ (since each evaluation preserves roABP size by \Cref{fact:roABP}). Furthermore, as $g$ can also be computed by a $\poly(n)$-size roABP, we see that 
    $$
    1-C(\xbar,0,g ,\wbar) = 1- C_0(\xbar,\wbar) -C_1(\xbar,\wbar) g
    $$
    has a $\poly(s,n)$-size roABP in the order of variables that $C$ induces on $\xbar, \wbar$. As each Boolean axiom $x_i^2- x_i$ only refers to a single variable, substituting $\wbar \leftarrow \xbar^2-\xbar$ for $1-C(\xbar, 0 ,g, \wbar)$ in the roABP will preserve obliviousness, but now each variable is read twice. Therefore, $1-C(\xbar, 0 ,g, \xbar^2- \xbar)$ has a $\poly(s,n)$-size read-twice oblivious ABP. Finally, using the fact that a nonzero multiple of $f$ requires $\exp(\Omega(n))$ size to be computed as read-twice oblivious ABPs (\Cref{lemma: rtABP lower bounds for multiple}), it follows that $\poly(s,n) \geq \exp(\Omega(n))$, implying $s \geq \exp(\Omega(n))$ as desired.
\end{proof}

\subsection{Limitations}

The following discusses the limitations of the functional lower bound method for roABP-$\IPS$. Namely, we show that it is impossible to get a non-placeholder functional lower bound against roABP-$\IPS$ over finite fields, even if the refutation is restricted to a multilinear polynomial. We first recall this fact about roABPs.

\begin{fact}
\label{fact:roabp_width}
If $f,g\in\mathbb{F}[\overline{x}]$ are computable by width-$r$ and width-$s$ roABPs respectively, then
\begin{itemize}
    \item $f+g$ is computable by a width-$(r+s)$ roABP.
    \item $f\cdot g$ is computable by a width-$(rs)$ roABP.
\end{itemize}
\end{fact}

Now, as discussed in \Cref{sec:upper_bound}, for a given unsatisfiable instance $f$ in finite field $\mathbb{F}_p$, by Fermat's Little Theorem we have the following refutation:  
\begin{equation}
\label{eq:finite_ub}
    f(\overline{x})^{p-2}f(\overline{x})=1 \mod \overline{x}^2-\overline{x}.
\end{equation}

Thus, if $f$ is easy for roABPs then by \Cref{fact:roabp_width}, so is $f(\overline{x})^{p-2}$ (as $p$ is constant), so in this case it is impossible to acheive a lower bound on roABP-$\IPS$ refutations. Now, consider the case where refutations must be multilinear (that is, an analogue to the constant-depth multilinear $\IPS$ proof system from \Cref{sec:cd-IPS}). In this proof system, the refutation in \Cref{eq:finite_ub} cannot work, as it is not multilinear. However, it is shown in \cite{FSTW21} that roABPs are closed under multilinearization. We restate their result for concreteness.

\begin{proposition}[Proposition 4.5 from \cite{FSTW21}]
    Let $f\in\mathbb{F}[\overline{x}]$ be computable by a width-$r$ roABP, in order of the variables $x_1<\cdots<x_n$, and with individual degrees at most $d$. Then, ml$(f)$ has a poly$(r,n,d)$-explicit width-$r$ roABP in order of the variables $x_1<\cdots<x_n$.
\end{proposition}
Thus, we simply consider the multilinear polynomial $g=\text{ml}(f^{p-2})$ to be our refutation (as $g$ agrees with $f^{p-2}$ over the Boolean cube, implying \Cref{eq:finite_ub} holds). By the above proposition, since $f^{p-2}$ has a small roABP computing it, so does $g$. This leads to the following theorem.

\begin{theorem}
\label{thm:flbm_limitation}
    The functional lower bound method cannot establish non-placeholder lower bounds on the size of roABP-$\IPS$ refutations when working in finite fields.
\end{theorem}

\section{Towards Lower Bounds for CNF Formulas}\label{sec:translation-lemmas}

    We now turn to the problem of establishing lower bounds for CNF formulas. In the previous sections, the lower bounds we presented were for algebraic instances. In contrast, we show that an $\IPS$ lower bound against an unsatisfiable set of one or more polynomial equations over finite fields implies the existence of a hard Boolean instance. However, this implication requires a subsystem of $\IPS$ which can reason with large degree, therefore our results do not meet this criteria. Accordingly, the existence of any hard instance for $\IPS$ over finite fields (even when the equations are given as algebraic circuits), allowing refutations of possibly exponential total degree, implies the existence of hard Extended Frege instances. Similarly, if the hard instance is only against $\IPS$ refutations of polynomial total degree, then there are hard instances against Frege. 
    \medskip 
    
    We work in finite field $\mathbb{F}_q$ where $q$ is a constant (independent of the size of the formulas and their number of variables). When we work with $\CNF$ formulas in $\IPS$ we assume that the $\CNF$ formulas are translated according to the following definition.
    \begin{definition}[Algebraic translation of CNF formulas]
        \label{def: Algebraic translation of CNF formulas}
        Given a $\CNF$ formula in the variables $\overline{x}$, every clause $\bigvee_{i \in P}x_i \lor \bigvee_{j \in N} \neg x_j$ is translated into $\prod_{i \in P}(1-x_i) \cdot \prod_{j \in N}x_j=0$. (Note that these expressions are represented as algebraic circuits, where the products are not expanded.)
    \end{definition}
    Notice that a $\CNF$ formula is satisfiable by 0-1 assignment if and only if the assignment satisfies all the equations in the algebraic translation of the $\CNF$. The following definitions are taken from \cite{ST25}, and we supply them here for completeness. 
    \begin{definition}[Algebraic extension axioms and unary bits \cite{ST25}] 
        \label{def: algebraic extension axioms and unary bits}
        Given a circuit $C$ and a node $g$ in $C$, we call the equation
        \begin{equation*}
            x_g = \sum_{i=0}^{q-1} i \cdot x_{g_i}
        \end{equation*}
        the algebraic extension axiom of $g$, with each variable $x_{g_i}$ being the $i$th unary-bit of $g$.
    \end{definition}
    
    \begin{definition}[Plain CNF encoding of constant-depth algebraic circuit $\cnf(C(\overline{x}))$ \cite{ST25}]
        \label{def: plain CNF encoding of constant-depth algebraic circuit cnf(C(x))}
        Let $C(\overline{x})$ be a circuit in the variables $\overline{x}$. The plain $\CNF$ encoding of the circuit $C(\overline{x})$, denoted $\cnf(C(\overline{x}))$ consists of the following $\CNF$s in the unary-bits variables of all the gates in $C$ and extra extension variables (and only in the unary-bit variables):
        \begin{enumerate}
            \item If $x_i$ is an input node in $C$, the plain $\CNF$ encoding of $C$ uses the variables $x_{x_{i0}},\cdots,x_{x_{i(q-1)}}$ that are the unary-bits of $x_i$, and includes clauses ensuring that exactly one unary bit is 1 and all others are 0:
            \begin{equation*}
                    \bigvee_{j=0}^{q-1} x_{x_i j} \land \bigwedge_{j \neq l \in \{0,\cdots,q-1\}} (\neg x_{x_i j} \lor \neg x_{x_i l}).
            \end{equation*}
            \item If $\alpha \in \mathbb{F}_q$ is a scalar input node in $C$, the plain $\CNF$ encoding of $C$ contains the $\{0,1\}$ constants corresponding to the unary-bits of $\alpha$. These constants are used when fed to (translation of) gates according to the wiring of $C$ in item 4.
            \item For every node $g$ in $C(\overline{x})$ and every satisfying assignment $\overline{\alpha}$ to the plain $\CNF$ encoding, the corresponding unary-bit $x_{g_i}$ evaluates to $1$ if and only if the value of $g$ equals $i \in {0,\dots, q-1}$ (when the algebraic inputs $\overline{x} \in (\mathbb{F}_q)^\ast$ to $C(\overline{x})$ take on the values corresponding to the Boolean assignment $\overline{\alpha}$; "$\ast$" here means the Kleene star). This is ensured with the following equations: if $g$ is a $\circ \in \{+, \times\}$ node that has inputs $u_1,\cdots, u_t$. Then we consider the following equations:
                \begin{align*}
                     &u_1 \circ u_2 = v_1^g \\
                    &u_{i+2} \circ v_{i}^g = v_{i+1}^g, \qquad 1 \leq i \leq t-3 \\
                    &u_{t} \circ v_{t-2}^g = g.
                \end{align*}
                Then, for each equation above, for simplicity, we denote as $x \circ y =z$. For each $x+y=z$ we have a $\CNF$ $\phi$ in the unary-bits variables of $x, y,z$ that is satisfied by assignment precisely when the output unary-bits of $z$ get their correct values based on the (constant-size) truth table of $\circ$ over $\mathbb{F}_q$ and the input unary-bits of $x$ and $y$ (we ensure that if more than one unary-bit is assigned $1$ in any of the unary-bits of  $x,y,z$ then the $\CNF$ is unsatisfiable).
            \item For every unary-bit variable $x_{g_i}$, we have the Boolean axiom (recall we write these Boolean axioms explicitly since we are going to work with $\IPS^\alg$):
            \begin{equation*}
                x_{g_i}^2 - x_{g_i} = 0.
            \end{equation*}
        \end{enumerate}
        
    \end{definition}
    Therefore, we can see that the formula size of $\cnf(C(\overline{x})=0)$ is $\poly(q^2 \cdot |C|)$.
    \begin{definition}[Plain CNF encoding of a constant-depth circuit equation $\cnf(C(\overline{x})=0)$ \cite{ST25}]
        \label{def: plain CNF encoding of a constant-depth circuit equation cnf(C(x)=0)}
         Let $C(\overline{x})$ be a circuit in the variables $\overline{x}$. The plain $\CNF$ encoding of the circuit equation $C(\overline{x})=0$ denoted $\cnf(C(\overline{x})=0)$ consists of the following $\CNF$ encoding from \Cref{def: plain CNF encoding of constant-depth algebraic circuit cnf(C(x))} in the unary-bits variables of all the gates in $C$ ( and only in the unary-bit variables), together with the equations:
         \begin{equation*}
                x_{g_{out}0} =1 \quad \text{and} \quad x_{g_{out}i} = 0, \quad \text{for all }i=1,\cdots,q-1,
            \end{equation*}
            which express that $g_{out} = 0$, where $g_{out}$ is the output node of $C$.
        
    \end{definition}
 
    \begin{definition}[Extended CNF encoding of a circuit
    equation (circuit, resp.); $\ecnf(C(\overline{x})=0)$ ($\ecnf(C(\overline{x}))$, resp.) \cite{ST25}]
        \label{def: Extended CNF encoding of a circuit equation cnf(C(x)=0)}
         Let $C(\overline{x})$ be a circuit in the variables $\overline{x}$ over the finite field $\mathbb{F}_q$. The extended $\CNF$ encoding of the circuit equation $C(\overline{x})=0$ (circuit $C(\overline{x})$, resp.), in symbols $\ecnf(C(\overline{x})=0)$ ($\ecnf(C(\overline{x}))$, resp.), is defined to be a set of algebraic equations over $\mathbb{F}_q$ in the variables $x_g$ and $x_{g0},\cdots,x_{gq-1}$ which are the unary-bit variables corresponding to the node $g$ in $C$, that consist of:
         \begin{enumerate}
             \item the plain $\CNF$ encoding of the circuit equation $C(\overline{x})=0$ (circuit $C(\overline{x})$, resp.), namely, $\cnf(C(\overline{x})=0)$ ($\cnf(C(\overline{x}))$, resp.); and
             \item the algebraic extension axiom of $g$, for every gate $g$ in $C$.
         \end{enumerate}
         
    \end{definition}

    Since we work with extension variables, it is more convenient to express a circuit equation $C(\overline{x})=0$ as a set of equations encoding each gate of $C$, which we call the straight line program of $C(\overline{x})$ (and is equivalent in strength to algebraic circuits).
    
    \begin{definition}[Straight line program ($\SLP$)]
        \label{def: Straight line program SLP}
        An $\SLP$ of a circuit $C(\overline{x})$, denoted by $\SLP(C(\overline{x}))$, is a sequence of equations between variables such that the extension variable for the output node computes the value of the circuit assuming all equations hold. Formally, we choose any topological order $g_1,g_2, \cdots,g_i,\cdots,g_{|C|}$ on the nodes of the circuit $C$ (that is, if $g_j$ has a directed path to $g_k$ in $C$ then $j < k$) and define the following set of equations to be the $\SLP$ of $C(\overline{x})$:
        \begin{center}
            $g_i = g_{j1} \circ g_{j2} \circ \cdots \circ g_{jt}$ for $\circ \in \{+, \times \}$ iff $g_i$ is a $\circ$ node in $C$ with $t$ incoming edges from $g_{j1},\cdots,g_{jt}$.
        \end{center}
        An $\SLP$ representation of a circuit equation $C(\overline{x})=0$ means that we add to the $\SLP$ above the equation $g_{|C|}=0$, where $g_{|C|}$ is the output node of the circuit.
    \end{definition}
    
    The following lemma, which we refer to as the \emph{translation lemma} throughout this paper, shows that we can derive (with some additional axioms) the circuit equation $C(\overline{x})=0$ given the extended $\CNF$ encoding of this circuit equation $\ecnf(C(\overline{x})=0)$, and vice versa.
    \begin{lemma}[Translating between extended $\CNF$s and circuit equations in Fixed Finite Fields \cite{ST25}] 
        \label{lemma: Translating between extended CNF formulas and circuit equations in Fixed finite fields}
        Let $\mathbb{F}_q$ be a finite field, and let $C(\overline{x})$ be a circuit of depth $\Delta$ in the $\overline{x}$ variables over $\mathbb{F}_q$. Then, both of the following hold:
        \begin{align}
            \ecnf(C(\overline{x})=0) \sststile{\IPS^\alg}{*,O(\Delta)} C(\overline{x})=0.  
        \end{align}
\begin{equation}
\hspace*{-6.5em}
\begin{aligned}
\begin{array}{@{}r@{}}
\{x_g  =\sum_{i=0}^{q-1} i \cdot x_{gi}: g \text{ is a node in } C\}, \\[0.7em]
\{x_{gi}^2 - x_{gi} =0: g \text{ is a node in } C, 0 \leq i < q\}, \\[0.7em]
\{\sum_{i=0}^{q-1} x_{gi} =1 : g \text{ is a node in } C\}, \\[0.7em]
\SLP(C(\overline{x})), \\[0.7em]
C(\overline{x})=0
\end{array}
\quad
\sststile{\IPS^\alg}{*,O(\Delta)}
\quad
\ecnf(C(\overline{x})=0).
\end{aligned}
\label{eq:your_label}
\end{equation}
    \end{lemma}
\smallskip

    \begin{proposition}[Proposition 3.7 in \cite{ST25}]
        \label{prop: C <-> cnf <-> ecnf fixed finite fields}
        Let $C(\overline{x})=0$ be a circuit equation over $\mathbb{F}_q$ where $q$ is any constant prime. Then, $C(\overline{x})=0$ is unsatisfiable over $\mathbb{F}_q$ iff $\cnf(C(\overline{x})=0)$ is an unsatisfiable $\CNF$ iff $\ecnf(C(\overline{x})=0)$ is an unsatisfiable set of equations over $\mathbb{F}_q$.
    \end{proposition}

    From here, we extend their result by eliminating these additional axioms in both directions. The only additional axioms we need are the \emph{field axioms} $\{x^q -x =0: x \text{ is a variable in } C\}$, which can be easily derived from the Boolean axioms if the variables in the circuit are Boolean (as we are working in finite fields). We use $\UBIT_j(x)$ to denote the following Lagrange polynomial:
    \begin{equation*}
        \label{UBITs}
        \frac{\prod_{i=0, i \neq j}^{q-1}(x-i)}{\prod_{i=0, i \neq j}^{q-1}(j-i)},
    \end{equation*}
    where $x$ can be a single variable or an algebraic circuit. Hence, it is easy to observe that
    \begin{equation*}
        \UBIT_j(x) = \begin{cases}
            1, \quad x = j,\\
            0, \quad \text{otherwise}.
        \end{cases}
    \end{equation*}
    Also, suppose $x$ has size $|x|$ and depth $\depth(x)$ (when $x$ is a single variable, it has size $1$ and depth $1$). Then, $\UBIT_j(x)$ can be computed by an algebraic circuit of size $O(|x|^{q-1})$ and depth $\depth(x)+2$. In addition, we introduce a Semi-CNF $\SCNF$, which is a substitution instance of a $\CNF$.

    \begin{definition}[Semi-CNF $\SCNF$ encoding of a constant-depth circuit equation $\SCNF(C(\overline{x})=0)$]
        \label{def: semi-CNF encoding}
        Let $C(\overline{x})$ be a circuit in the variables $\overline{x}$. The semi-CNF encoding of the circuit equation $C(\overline{x})=0$ denoted $\SCNF(C(\overline{x}))$ is a substitution instance of the plain CNF encoding in \Cref{def: plain CNF encoding of a constant-depth circuit equation cnf(C(x)=0)} where each unary-bits $x_{uj}$ of all the gates and extra extension variables $u$ is substituted with $\UBIT_j(C_u)$ where $C_u$ is the constant-depth algebraic circuit computes $u$.\footnote{This $C_u$ can be constructed from $\SLP$s easily.}
    \end{definition}
    We now demonstrate the connection between semi-CNFs and circuit equations.

    \begin{theorem}[Translate semi-CNFs from circuit equations in Fixed Finite Fields]
        \label{lemma: Translate semi-CNFs from circuit equations in Fixed Finite Fields}
        Let $\mathbb{F}_q$ be a finite field, and let $C(\overline{x})$ be a circuit of depth $\Delta$ in the $\overline{x}$ variables over $\mathbb{F}_q$. Then, the following holds
        \begin{gather*}
            \{x^q -x =0: x \text{ is a variable in } C\}, \;C(\overline{x})=0  \sststile{\IPS^\alg}{*,O(\Delta)} \SCNF(C(\overline{x})=0).
        \end{gather*}
    \end{theorem}

Since the field equations $x^q-x=0$ are efficiently derivable from the Boolean axioms, we get the following for $\IPS$ (which by default contains the Boolean axioms):
\begin{gather*}
        C(\overline{x})=0  \sststile{\IPS^\alg}{*,O(\Delta)} \SCNF(C(\overline{x})=0).
        \end{gather*}

    \begin{proof}
        In \Cref{lemma: Translating between extended CNF formulas and circuit equations in Fixed finite fields}, the given axioms include:
        \begin{enumerate}
        \centering
            \item[$(i)$] $\{x_g = \sum_{i=0}^{q-1}i \cdot x_{gi}: g \text{ is a node in } u\circ v=w \}$,
            \item[$(ii)$] $\{x_{gi}^2 - x_{gi} = 0: g \text{ is a node in } u\circ v =w\}$,
            \item[$(iii)$] $\{\sum_{i=0}^{q-1} x_{gi} =1 : g \text{ is a node in } u\circ v = w\}$,
            \item[$(iv)$] $\sum_{i=0}^{q-1}i \cdot x_{ui} \circ \sum_{i=0}^{q-1}i \cdot x_{vi} = \sum_{i=0}^{q-1}i \cdot x_{wi}$,   
         \end{enumerate}
        and there is a constant-depth constant-size $\IPS$ derivation of the plain CNF encoding of $u \circ v = w$. Thus, we must show that we can derive the above four axioms when we substitute $x_{gi}$ with $\UBIT_i(C_g)$. Due to the standard property of Lagrange polynomials, the following circuit equation is a polynomial identity, which can be proved freely in $\IPS$ (in finite field $\mathbb{F}_q$):
        \begin{equation*}
            x = \sum_{j=0}^{q-1} j \cdot \UBIT_j(x),
        \end{equation*}
        which is exactly the axiom in $(i)$. Hence, we know that $C_u = \sum_{j=0}^{q-1} j \cdot \UBIT_j(C_u)$, $C_v = \sum_{j=0}^{q-1} j \cdot \UBIT_j(C_v)$ and $C_u \circ C_v = C_w = \sum_{j=0}^{q-1} j \cdot \UBIT_j(C_w)$. These polynomial identities give us the substitution instance of the last equation $(iv)$:
        \[
             \sum_{j=0}^{q-1} j \cdot \UBIT_j(C_u) \circ  \sum_{j=0}^{q-1} j \cdot \UBIT_j(C_v) =  \sum_{j=0}^{q-1} j \cdot \UBIT_j(C_w).
        \]

        The second set of equations $(ii)$ is in the ideal of the field axioms for $g$. We show that in depth-$O(\Delta)$ and polynomial size, we can derive the field axioms $C_g^q - C_g =0$ for all circuits that compute the nodes and extension variables (using the field axioms $x^q-x=0$, for every input variable). We derive the field axioms for nodes and extension variables by induction on depth.  
        When $g$ is a $\circ \in \{+, \times\}$ node that has inputs $u_1,\cdots, u_t$, the $\SLP$s includes:
        \begin{align*}
            &u_1 \circ u_2 = v_1^g \\
            &u_{i+2} \circ v_{i}^g = v_{i+1}^g, \qquad 1 \leq i \leq t-3 \\                &u_{t} \circ v_{t-2}^g = g.
        \end{align*}
        For each $v_i^{g} = u_1 \circ \cdots \circ u_{i+2}$, $C_{v_i^g} =C_{u_1} \circ \cdots \circ C_{u_{i+2}}$ is a polynomial identity. By induction, we already have the field axioms for all $C_{u_i}$. We show that we can derive the field axioms for all $C_{v_i^g}$ and $C_g$ simultaneously. Now suppose $\circ = +$, then the following equations hold over $\F_q$:
        \begin{align*}
            C_{v_i^g}^q &\equiv (C_{u_1}+ \cdots + C_{u_{i+2}})^q \\
                        & \equiv C_{u_1}^q + \cdots + C_{u_{i+2}}^q \\
                        & \equiv C_{u_1} + \cdots + C_{u_{i+2}} \\
                        & \equiv C_{v_i^g}.
        \end{align*}
        The proof for the node $g$ is the same. We can therefore conclude that if $\circ =+$, we can derive the field axioms for all $C_{v_i^g}$ and $C_g$ simultaneously. Suppose $\circ = \times$, then the following equations hold over $\F_q$:
        \begin{align*}
            C_{v_i^g}^q &\equiv (C_{u_1}\times \cdots \times C_{u_{i+2}})^q \\
                        & \equiv C_{u_1}^q \times \cdots \times C_{u_{i+2}}^q \\
                        & \equiv C_{u_1} \times \cdots \times C_{u_{i+2}} \\
                        & \equiv C_{v_i^g}.
        \end{align*}
        Again, the proof for the node $g$ is the same and thus we conclude that given the field axioms for the input variables,  we can derive the field axioms for all circuits that compute the nodes and extension variables in depth $O(\Delta)$ and polynomial size. It remains to show that $\UBIT_j(x)^2 - \UBIT_j(x) =0$ is in the ideal of the field axiom of $x$, for any $x$. The equation $x^q-x=0$ is the unique monic polynomial of degree $q$ that has all elements of $\mathbb{F}_q$ as roots. Therefore, any polynomial $f(x) \in \mathbb{F}_q[x]$ that vanishes when evaluated to any $x \in \mathbb{F}_q$ must be divisible by $x^q-x$. It is easy to check that $\UBIT_j(x)^2 - \UBIT_j(x) $ vanishes at all points, implying it is in the ideal generated by $x^q-x$. Hence, there is a degree (of $x$) $q-2$ polynomial $Q(x)$ such that $Q(x) \cdot (x^q-x) = \UBIT_j(x)^2 - \UBIT_j(x)$, and as a result there is a depth-$\Delta$ polynomial-size proof for $\UBIT_j(x)^2 - \UBIT_j(x) =0$ from $x^q-x$.

        Finally, $\sum_{j=0}^{q-1} \UBIT_j(x) =1$ is a polynomial identity, for every $x$. This follows from the fact that it is a single-variable polynomial with degree $q-1$, but has $q$ many distinct roots. By the fundamental theorem of algebra, it must be a zero polynomial, and consequently we get axioms in $(iii)$ for free in $\IPS$. All together, we can conclude that
        \[
            \{x^q -x =0: x \text{ is a variable in } C\}, C(\overline{x})=0  \sststile{\IPS^\alg}{*,O(\Delta)} \SCNF(C(\overline{x})=0
        \]
        by first deriving the substitution instance above, and then substituting the derivation for the CNF to get the derivation for the semi-CNF.
    \end{proof}

    Since $\SCNF$s are substitution instances of $\CNF$s, lower bounds for $\CNF$s imply lower bounds for $\SCNF$s, which gives the following corollary.
    \begin{corollary}[Lower bounds for circuit equations imply lower bounds for CNFs]
        \label{corollary: lower bounds for circuit equation implies lower bounds for CNFs}
         Let $\mathbb{F}_q$ be a finite field, and let $\{C(\overline{x})\}$ be a set of circuits of depth at most $\Delta$ in the Boolean variable $\overline{x}$. Then, if a set of circuit equations $\{C(\overline{x}) =0\}$ cannot be refuted in $S$-size, $O(\Delta^\prime)$-depth $\IPS$, then the CNF encoding of the set of circuit equations $\{\CNF(C(\overline{x})=0)\}$ cannot be refuted in ($S-\poly(|C|)$)-size, $O(\Delta^\prime + \Delta)$-depth $\IPS$.
    \end{corollary}

    \begin{lemma}[Translate circuit equations from semi-CNFs in fixed finite fields]
        \label{lemma: Translate circuit equations from semi-CNFs in Fixed Finite Fields}
        Let $\mathbb{F}_q$ be a finite field, and let $C(\overline{x})$ be a circuit of depth $\Delta$ in the $\overline{x}$ variables over $\mathbb{F}_q$. Then, the following holds:
        \begin{gather*}
            \{x^q -x =0: x \text{ is a variable in } C\}, \; \SCNF(C(\overline{x})=0)  \sststile{\IPS^\alg}{*,O(\Delta)} C(\overline{x})=0.
        \end{gather*}
    \end{lemma}

    \begin{proof}
        From the CNF encoding of each $\SLP$ axiom $u\circ v =w$ and the Boolean axioms for each unary bit, we have
        \[
            \sum_{j=0}^{q-1} j\cdot x_{uj} \circ \sum_{j=0}^{q-1} j\cdot  x_{vj} = \sum_{j=0}^{q-1} j\cdot  x_{wj}
        \]
        in constant-depth polynomial-size $\IPS$. As we showed in the proof of \Cref{lemma: Translate semi-CNFs from circuit equations in Fixed Finite Fields}, the field axioms for all circuits that compute nodes and extension variables can be derived from the field axioms of the input variables, in constant-depth polynomial-size $\IPS$. We also showed that these derived field axioms can in turn also derive the Boolean axioms for all $\UBIT$ polynomials (of circuits that compute nodes and extension variables) in constant-depth polynomial-size $\IPS$. As a result, substituting each $x_{gj}$ above with $\UBIT_j(C_g)$ for $g \in \{u,v,w\}$ and $0 \leq j \leq q-1$, we get a constant-depth polynomial-size $\IPS$ derivation of
        \[
            \sum_{j=0}^{q-1} j\cdot  \UBIT_j(C_u) \circ \sum_{j=0}^{q-1} j \cdot \UBIT_j(C_v) = \sum_{j=0}^{q-1} j \cdot \UBIT_j(C_w)
        \]
        from the semi-CNF encoding of $u \circ v =w$ and the Boolean axioms for each $\UBIT$. Lastly, as mentioned in the proof of  \Cref{lemma: Translate semi-CNFs from circuit equations in Fixed Finite Fields}, in finite field $\mathbb{F}_q$ we get the following circuit equation for free in $\IPS$ (as it is a polynomial identity):
        \begin{equation*}
            x = \sum_{j=0}^{q-1} j \cdot \UBIT_j(x).
        \end{equation*}

        Therefore, we get the full $\SLP$ for the circuit equation $C(\overline{x})=0$, and consequently the circuit equation can easily be obtained from this $\SLP$.
    \end{proof}


\section*{Acknowledgment}
We are thankful to Rahul Santhanam for helpful discussions on the translation lemmas (\Cref{sec:translation-lemmas}) and for \citeallauthors{BLRS25} for letting us know about their independent work.

\renewcommand{\mathbf}[1]{#1} 

\small 
\printbibliography

\end{document}